\documentclass[12pt]{article}    
\usepackage{amssymb, amsmath, amsthm, mathtools} 
\usepackage{graphicx, dsfont, enumerate, color, setspace} 
\usepackage{url}
\usepackage{caption}
\usepackage{subcaption}
\usepackage{outline}
\usepackage{adjustbox}

\usepackage[colorlinks = true,
            linkcolor = blue,
            urlcolor  = blue,
            citecolor = blue,
            anchorcolor = blue]{hyperref}

\usepackage{authblk}

\usepackage{fullpage} 

\usepackage{natbib}
\bibpunct{(}{)}{;}{a}{,}{,}

\usepackage{titlesec}
\titleformat{\subsubsection}[runin]
{\normalfont\em}{\thesubsubsection}{1em}{}
\newcommand{\subsub}[1]{\subsubsection{#1.}}

\newtheorem{remark}{Remark}
\newtheorem{lemma}{Lemma}
\newtheorem{theorem}{Theorem}
\newtheorem{proposition}{Proposition}
\newtheorem{assumption}{Assumption}
\newtheorem{definition}{Definition}

\newtheorem{todo}{Todo}

\DeclareMathOperator*{\argmax}{argmax}
\DeclareMathOperator*{\argmin}{argmin}
\newcommand{\one}{\mathbf{1}}
\newcommand{\zero}{\mathbf{0}}
\newcommand{\be}{\textbf{e}}
\newcommand{\bZ}{\textbf{Z}}
\newcommand{\bh}{\textbf{h}}
\newcommand{\mZ}{\underline{\textbf{Z}}}
\newcommand{\bY}{\textbf{Y}}
\newcommand{\by}{\textbf{y}}
\newcommand{\mY}{\underline{\textbf{Y}}}
\newcommand{\mbY}{\underline{\mathds{Y}}}

\newcommand{\Yset}{\mathcal{Y}}

\newcommand{\Cset}{\mathcal{C}}
\newcommand{\Tset}{\mathcal{T}}
\newcommand{\iv}{\mathds{1}}
\newcommand{\Zit}{Z_{it}}
\newcommand{\Zitall}{Z_{i, 1:t}}

\newcommand{\Pulses}{\mathcal{E}}
\newcommand{\Wedges}{\mathcal{E}^\ast}
\newcommand{\PulsesSupport}{\mathcal{Z}(\mathcal{E})}
\newcommand{\WedgesSupport}{\mathcal{Z}(\mathcal{E}^\ast)}
\newcommand{\PulsesDesign}{\mathds{H}}
\newcommand{\WedgesDesign}{\mathds{H}^\ast}

\newcommand{\bI}{\mathds{I}}
\newcommand{\bw}{\textbf{w}}

\newcommand{\Sset}{\mathsf{S}}

\newcommand{\bx}{\textbf{x}}
\newcommand{\byy}{\textbf{y}}
\newcommand{\bz}{\textbf{z}}

\newcommand{\ATE}{\mathrm{ATE}}

   \newcommand{\dyn}{habituation }
   \newcommand{\dyneff}{\dyn effect}
      \newcommand{\dyneffs}{\dyn effects}

 \title{Minimax designs for causal effects \\ in temporal experiments 
 with treatment habituation}

\author[1]{Guillaume Basse}
\affil[1]{Stanford University, Department of MS\&E and Department of Statistics}
\author[2]{Yi Ding}
\affil[2]{University of Chicago, Department of Computer Science}
\author[3]{Panos Toulis}
\affil[3]{University of Chicago, Booth School of Business}

\date{\today}

\begin{document}

\maketitle


\begin{abstract}
  Randomized experiments are the gold standard for estimating the causal effects
  of an intervention. In the simplest setting, each experimental unit is randomly
  assigned to receive treatment or control, and then the outcomes in each treatment
  arm are compared. 
  In many settings, however, randomized experiments need to be executed over 
  several time periods
  such that treatment assignment happens at each time period. 
  In such temporal experiments, it has been observed that the effects of an
  intervention on a given unit may be large when the unit is first exposed to it,
  but then it often attenuates, or even vanishes, after repeated exposures. 
This phenomenon is typically due to units' habituation to the intervention, 
or some other general form of learning, such as when users gradually start to ignore repeated mails sent by a promotional campaign.
This paper proposes randomized designs for estimating causal effects in
temporal experiments when habituation is present. We show that our 
designs are minimax optimal in a large class
of practical designs. Our analysis is based on the randomization framework of causal inference, and imposes no parametric modeling assumptions on  the outcomes.
  
\end{abstract}

\newpage
\small
\tableofcontents


\onehalfspacing

\normalsize
\section{Introduction}
\label{sec:intro}

In many causal inference problems, it is of interest to understand how treatment 
effects vary over time. This concerns, for example, the evaluation of the long-term impacts
of public policies~\citep{athey2019surrogate, allcott2014short, imai2011experimental, sjolander2016carryover, hainmueller2014causal}, or applications where treatment effects may carry over from earlier intervention periods, as in clinical trials~\citep{copas2015designing, wellek2012proper, brown1980crossover}. In this context, the main challenge is that a unit's current outcome may be affected not just by its current treatment, but also by its past treatments. 

For instance, in a recent paper, \cite{allcott2014short} studied the impact of
mailing a customized energy report to households across the United States on their energy consumption. They repeated this treatment every month on the same households over multiple years, and recorded their monthly energy consumption over the same period. One question at the center of their study was the effect of this repeated treatment,
and in particular whether households had habituated to the treatment.
Indeed, the authors observed that while households' energy consumption
was reduced upon receiving the first monthly report, the magnitude of the effect
shrank dramatically after receiving just a few more reports. Other experiments
studying the impact of behavioral interventions for energy
conservation have made similar observations \citep{abrahamse2005review,
  hahn2016impact, allcott2010behavior}. In fact, this phenomenon of habituation
after repeated exposure to an intervention has been reported in a variety of
contexts, including online advertising
\citep{chatterjee2003modeling, hohnhold2015focusing, yan2019measuring}, marketing
\citep{wathieu2004consumer, liberali2011effects}, and ergonomics
\citep{kim2009habituation}. For instance, in the context of online advertising,
\cite{hohnhold2015focusing} ran a number of experiments to measure
``ads blindness", which is a phenomenon describing the behavior of online 
users who, when repeatedly exposed to low-quality ads, progressively learn to
ignore them. This paper considers the problem of designing optimal experiments
for measuring effects of this kind.

Several experimental designs have been applied in practice to quantify habituation
under repeated treatments.~\cite{hohnhold2015focusing} proposed two designs –––
the ``Post-Period design" and the ``Cookie--Cookie-Day Design" ––– for
capturing and assessing this phenomenon. The literature on behavioral
interventions (e.g.,~\cite{allcott2014short}) also uses a variety of designs with
similar ideas. 
However, the design question is not studied formally in this literature, and
optimality properties are not discussed.

\newpage
In contrast, design optimality is at the heart of the literature
on sequential personalized experiments~\citep[for example]{lei2012smart},
where the objective is to adaptively construct a sequence of treatments
for each experimental unit in order to maximize an outcome of interest~(e.g. blood
pressure, quality of life), typically based on covariates and outcomes at
previous stages. 
That setting, however, differs from ours in two fundamental
ways. First, the goal of maximizing unit responses is distinct from, and
in fact it may conflict with, the goal of estimating temporal causal effects.
Second, the literature on adaptive sequential experiments generally requires strong modeling assumptions, whereas our
approach makes no modeling assumption on the unit outcomes. 
Instead, we make use of the randomization framework in causal inference where the potential outcomes are fixed, and randomness comes only from the experimental design.

The goal of this paper is to formalize and develop minimax optimal designs for
estimating causal effects in temporal experiments where each unit can be exposed
to a different treatment at each time period, and the resulting outcomes are also
observed for all units at each time period.
%
We adopt the potential outcomes framework of causal inference~\citep{rubin1974estimating, neyman1923application}, and the
randomization-based perspective on inference, extending both to the temporal
setting.
Methodologically, our setup is therefore closer to that of~\citet{bojinov2019time} who also consider dynamic potential outcomes and causal effects in temporal experiments. The key difference is that~\citet{bojinov2019time} consider Fisherian randomization tests for hypotheses on the potential outcomes, whereas we address the design problem of minimax estimation of causal effects.
Our study of minimax optimality builds
on early work by \cite{wu1981minimax} and \cite{li1983minimaxity} who showed
that in the cross-sectional setting, the balanced completely randomized design is
minimax optimal. We extend their results to the temporal setting, but find that
the minimax optimal design is no longer balanced among all arms --- we provide
some intuition as to why the symmetry present in the cross-sectional setting is
broken in the temporal setting.

The rest of this paper is structured as follows. In Section~\ref{sec:setup}, we describe our
framework in detail, and identify two types of temporal causal
estimands of interest: {\em \dyneffs}, as described above, and {\em instantaneous effects},
which measure the effect of a single exposure to the treatment at a given time
period. We develop the main theory in Section~\ref{sec:minimax}, and introduce
minimax optimal designs for jointly estimating \dyneff\ and instantaneous
effects. We discuss extensions of our basic design under alternative
sets of assumptions in Section~\ref{sec:extensions}, and illustrate our results
with simulations in Section~\ref{sec:simulations}.

\section{Setup and notation}
\label{sec:setup}

\subsection{Temporal experiments}
\label{sec:designs}
%
We consider a temporal experiment with $N$ units, which is taking place over $T$ discrete periods,
indexed by $t=1, \ldots, T$. We assume that $T$ is fixed and known. At each time $t$, a unit $i$
can either receive treatment or control, denoted by $\Zit=1$ or $\Zit = 0$,
respectively. Vector
$\Zitall = (Z_{i1}, \ldots, \Zit)\in\{0, 1\}^t$ denotes the treatment history of $i$
up to time $t$. 
The full treatment history for unit $i$ is denoted by $\bZ_i = (Z_{i1}, \ldots, Z_{iT})$,
where the dependence on $T$ is left implicit; $\bZ \in \{0,1\}^{N\times T}$
denotes the $N\times T$ matrix of population treatment assignments, whose $i$-th
row is equal to $\bZ_i$. A temporal design is thus a distribution on $\bZ$.

In our paper, certain assignment vectors play a special role, which requires
additional notation. Let $\one = (1, \ldots, 1)$ and $\zero = (0, \ldots, 0)$ be, respectively, the {\em always-treated} and {\em always-control} assignment
vectors of length $T$; as before, dependence on $T$ is left implicit to simplify notation. We call {\em pulse assignment at time $t$}, and denote by
$\be_t$, the assignment vector that treats a unit only at time $t$; that is,
\begin{equation*}
\be_t = (0, \ldots, \underbrace{1}_{\text{index}~t}, 0, \ldots,0).
\end{equation*}
Similarly, the {\em wedge assignment at time $t$}, denoted $\bw_t$, treats a unit continuously 
after $t$:
\begin{equation*}
  \bw_t = (0, \ldots, \underbrace{1}_{\text{index}~t}, \ldots,1)
= \sum_{s=t}^T \be_s.
\end{equation*}
The assignment vectors $\one$, $\zero$, $\{\be_t\}_{t=1}^T$ and $\{\bw_t\}_{t=1}^T$
are the building blocks for the following two classes of designs that will be our main focus of analysis.
\begin{definition}[Pulse and wedge designs]
  Let $\Pulses=\{\one, \zero, \be_1, \ldots, \be_T\}$.
  Denote by
  $\PulsesSupport = \{\bZ \in \{0,1\}^{N\times T}: \bZ_i \in \Pulses,~
  \forall i=1,\ldots, N\}$ the set of all assignment matrices for which every unit's
  assignment is in $\Pulses$. Similarly, let 
  $\Wedges=\{\one, \zero, \bw_1, \ldots, \bw_T\}$ and denote by 
  $\WedgesSupport = \{\bZ \in \{0,1\}^{N\times T}: \bZ_i \in \Wedges,~
  \forall i=1,\ldots, N\}$ the set of all assignment matrices for which every unit
  assignment is in $\Wedges$. 
  A pulse design is a probability distribution, $\eta(\bZ)$, with
  support on $\PulsesSupport$. A wedge design is a probability
  distribution, $\eta_{\mathrm{w}}(\bZ)$, with support on $\WedgesSupport$. 
  The sets of all possible pulse designs and wedge designs are denoted by
  $\PulsesDesign$ and $\WedgesDesign$, respectively.
\end{definition}

Pulse and wedge designs are the building blocks of our causal estimands
(see Section~\ref{sec:estimands}). They are natural classes
of designs to consider, as they arise in many applications other than the ones we have considered so far;
for instance, wedge designs are popular in clinical trials
\citep{brown2006stepped, prost2015logistic,hargreaves2015five}. In online advertising, the Cookie--Cookie-Day design~\citep{hohnhold2015focusing} is an example of a wedge design.

In the following section, we focus on pulse designs because they are conceptually
simpler, but our key results hold unchanged for wedge designs as well~(see Remark~\ref{remark:wedges}).



%

\subsection{Potential outcomes and estimands}
\label{sec:po}
To define our causal estimands, we adapt the classical potential outcomes framework of causal
inference~\citep{neyman1923application,rubin1974estimating} to the temporal
setting, building on recent related work~\citep{ji2017randomization,bojinov2019time}.
In this framework, the
potential outcomes are fixed, and randomness
comes exclusively from the random assignment $\bZ$.

\subsub{Notation and assumptions}
\label{sec:po_notation}

Let $Y_{it}(\bZ)$ denote the scalar potential outcome of unit $i$ at time period $t$
under population assignment $\bZ$. 
We also denote by
$\bY_i(\bZ) = (Y_{i1}(\bZ), \ldots, Y_{iT}(\bZ))$ the vector of outcomes over time under
$\bZ$ for unit $i$, and by $\bY(\bZ) = [\bY_1(\bZ), \ldots, \bY_N(\bZ)]^\top$ the
$N\times T$ matrix whose $i$-th row is equal to $\bY_i(\bZ)$. We make the following two
assumptions.
\begin{assumption}[No  interference]\label{a:no-if}
  The outcome of unit $i$ at time $t$ depends only on it's own sequence of assignments.
  That is, for every unit $i$,
  \begin{equation*}
   \bZ_i = \bZ'_i~
    \Rightarrow~
    Y_{it}(\bZ) = Y_{it}(\bZ'),~\text{for all}~\bZ, \bZ',~\text{and all}~t=1, \ldots, T.
  \end{equation*}
\end{assumption}
\begin{assumption}[Non-anticipating outcome]\label{a:no-anticipating}
  The outcome of unit $i$ at time $t$ depends only on the treatment history up to
  time $t$. 
\end{assumption}

Assumption~\ref{a:no-if} extends to the temporal setting the classical
``no interference" assumption \citep{cox1958planning}, which is standard in causal
inference and econometrics~\citep{imbens2015causal}. Assumption~\ref{a:no-anticipating}
does not allow the potential outcomes at $t$ to depend on future treatments after $t$~\citep{bojinov2019time}, and is standard in  classical econometric studies ~\citep{heckman2005structural, heckman2016dynamic, robins1997causal, toh2008causal}. 
This assumption could fail, for example, in settings where experimental units are made aware of
(and can therefore anticipate) the future sequence of treatments. In this setting, the units' outcomes may reflect not only the effect of past treatments, but also the units' expectations of future treatments, 
complicating the analysis.

\subsub{Estimands}
\label{sec:estimands}
In the potential outcomes framework, the estimands of interest are defined as contrasts
of potential outcomes \citep{rubin1974estimating}.
In our setting, a simple causal estimand is the average treatment effect at time $t$,
namely,
\begin{equation}\label{eq:ATE1}
  \ATE(t) = \frac{1}{N} \sum_{i=1}^N Y_{it}(\one) - Y_{it}(\zero),
\end{equation}
which contrasts the potential outcomes at time $t$ between the always-treated
and always-control population assignments. This estimand entangles the effects
of past treatments with the effect of the current treatment, that is, 
it is a combination of {\em \dyneffs} and {\em instantaneous effects}. Formally, 
we can write:
\begin{flalign}
\ATE(t)
  &= \frac{1}{N}
  \sum_{i=1}^N \{Y_{it}(\one) - Y_{it}(\be_t) + Y_{it}(\be_t) - Y_{it}(\zero)\} \notag
  \\
  &= \frac{1}{N} \sum_{i=1}^N \{Y_{it}(\one) - Y_{it}(\be_t)\}  + \frac{1}{N} \sum_{i=1}^N \{Y_{it}(\be_t) - Y_{it}(\zero)\}.
  \label{eq:composite}
\end{flalign}
This decomposition leads to the following definitions.
\begin{definition}\label{def:estimands}
For a given time period $t$, the \dyneff, $\lambda_t$, and the instantaneous effect, $\delta_t$, are defined as:
\begin{equation}
\lambda_t   = \frac{1}{N} \sum_{i=1}^N \{Y_{it}(\one) - Y_{it}(\be_t)\}, \qquad
\delta_t  = \frac{1}{N} \sum_{i=1}^N \{Y_{it}(\be_t) - Y_{it}(\zero)\}.
\end{equation}
\end{definition}
As defined, the \dyneff, $\lambda_t$, contrasts the potential outcomes
under the always-treated assignment and the pulse assignment at time period $t$. These two
assignments have the same treatment at time $t$, so $\lambda_t$ intuitively
captures the treatment effect that can be attributed only to a cumulative effect from 
treatment history and not to treatment application at time $t$.
With this definition, we aim to capture the \dyneffs\ studied in the literature~\citep{hohnhold2015focusing, allcott2014short}. 
The instantaneous effect, $\delta_t$, contrasts the potential outcomes under the pulse assignment and the
always-control assignment. These two assignments have the same treatment history
up to time $t-1$ but differ in their treatment at time $t$. 
Thus, $\delta_t$ captures the effect of the treatment that can be attributed 
only to its application at time $t$, and not to any carryover effects from past treatments.
For the rest of this paper, our goal  will be to propose minimax optimal designs for estimating
the habituation and instantaneous effects jointly.

\begin{remark}
  The definitions of the \dyn\ and instantaneous effects involve only the potential
  outcomes $Y_{it}(\one)$, $Y_{it}(\zero)$ and $Y_{it}(\be_t)$; these potential outcomes
  derive from the assignments $\one$, $\zero$ and $\{\be_t\}$, which form the backbone
  of (and therefore justify the use of) the pulse design.
\end{remark}

\begin{remark}
  In the context of online advertisement, interest in the \dyneff\ has
  been motivated by the goal of estimating the long-term causal effects of an
  intervention based on a short-term experiment \citep{hohnhold2015focusing}. We
  believe that even if the ultimate objective is to estimate the long-term
  average treatment effect $\ATE(T^\ast)$ for $T^\ast > T$, one should still focus
  on the estimands $\{\lambda_t\}_{t=1}^T$ and $\{\delta_t\}_{t=1}^T$. Indeed, we
  argue that extrapolation of the $\ATE$ should be driven mostly by an extrapolation
  of the \dyneffs\ $\{\lambda_t\}_{t=1}^T$ which, in the context of online
  advertisement, can be supported by behavioral game theory
  \citep{chatterjee2003modeling, toulis2016long}. In contrast, the instantaneous effects
  $\{\delta_t\}_{t=1}^T$ are largely conjectural, and it is difficult to justify
  their extrapolation. In this context, our recommendation is to first
  estimate both the \dyn and instantaneous effects, then to verify that the
  instantaneous effects are stationary, and, finally, to extrapolate the \dyneffs.
  We leave this for future work.
\end{remark}

\newcommand{\hatl}{\widehat\lambda}
\newcommand{\hatd}{\widehat\delta}
\newcommand{\opt}{\mathrm{opt}}
\section{Minimax designs}
\label{sec:minimax}


\subsection{Estimators and risk}
\label{sec:estimators}

To state our minimax results, we need to define estimators for our
estimands, $\delta_t, \lambda_t$. To estimate $\lambda_t$ we define
the following plug-in estimator,
\begin{equation}\label{eq:hatl}
\hatl_t= \frac{1}{N_1} \sum_{i=1}^N \iv(\bZ_i = \one) Y_{it}(\one) -
\frac{1}{N_{e_t}} \sum_{i=1}^N \iv(\bZ_i = \be_t) Y_{it}(\be_t),~~t=2, \ldots, N,
\end{equation}
where
$N_1= \sum_{i=1}^N \iv(\bZ_i = \one)$ is the number of always-treated units, and
$N_{e_t} = \sum_{i=1}^N \iv(\bZ_i = \be_t)$ is the number of units assigned to a pulse
at time $t$. Variables $N_1, N_{e_t}$ are random because $\bZ$ is random in the experiment. Similarly, define the plug-in estimator of $\delta_t$ as follows,
\begin{equation}\label{eq:hatd}
\hatd_t = \frac{1}{N_{e_t}} \sum_{i=1}^N \iv(\bZ_i = \be_t) Y_{it}(\be_t) - 
\frac{1}{N_\zero} \sum_i^N \iv(\bZ_i = \zero) Y_{it}(\zero),~~t=2, \ldots, N,
\end{equation}
where
$N_\zero = \sum_{i=1}^N \iv(\bZ_i = \zero)$ is the number of always-control units. 
These plug-in estimators are simple and have well-studied sampling properties under
randomization~\citep{imbens2015causal}. In addition, their symmetry makes them amenable
to minimax analysis.

The risk of these estimators is a function of the design and  potential
outcomes. Let 
$\bY(\bz) = [\bY_1(\bz)~\ldots~\bY_N(\bz)]^\top$ denote the $N\times T$ matrix of potential outcomes under 
population assignment 
$\bz$,
whose $i$-{th} row is equal to $\bY_i(\bz)$. The full {\em schedule 
of
potential outcomes}, denoted by $\mY = [\bY(\zero), \bY(1), \bY(\be_2), \ldots, \bY(\be_T)]$,
therefore contains all the information needed for causal inference, since the causal estimands,
$\lambda_t$ and $\delta_t$, are deterministic functions of $\mY$. For a random assignment
$\bZ \in \PulsesSupport$ and a schedule of potential outcomes $\mY$, we consider the
squared loss function, 
\begin{equation}\label{eq:loss}
L(\bZ, \mY) = \sum_{t=2}^T (\hatl_t - \lambda_t)^2 + \sum_{t=2}^T(\hatd_t - \delta_t)^2,
\end{equation}
which is a function only of $\bZ$ and $\mY$, since $\hatl_t$ and $\hatd_t$ are
deterministic functions of $\bZ$ and $\mY$, while $\lambda_t$ and $\delta_t$ are 
functions of $\mY$ only. The risk of a pulse design, $\eta \in \PulsesDesign$, is then
defined as
\begin{equation}\label{eq:risk}
  r(\eta; \mY) = E_\eta\{L(\bZ, \mY)\},
\end{equation}
where the expectation is with respect to the randomization distribution in the design, $\eta(\bZ)$. 

As mentioned in Section~\ref{sec:po}, the schedule of potential outcomes,
$\mY$, is considered fixed, all randomness coming from $\bZ$. 
A minimax pulse design, $\eta^{\opt}$, minimizes the maximum risk over the support of
potential outcome schedules, denoted by $\mbY$, over  the entire class of pulse designs $\PulsesDesign$, that is, we define
\begin{equation}\label{eq:minimax}
\eta^{\opt} = \argmin_{\eta \in \PulsesDesign} \max_{\mY \in \mbY} \{r(\eta, \mY)\}.
\end{equation}
The task of obtaining a minimax design can therefore be seen as a game
between the statistician who chooses a pulse design $\eta \in \PulsesDesign$, and nature who chooses the worst-case schedule of potential outcomes from $\mbY$. 
To make progress, we impose an invariance property on $\mbY$, which 
we describe in the following section.

\subsection{Permutation invariance of potential outcomes}
\label{sec:perm}
In the cross-sectional setting, \cite{wu1981minimax}
considered minimax designs over permutation-invariant sets of model parameters to
reflect the experimenter's ignorance at the design stage. We adapt this idea to our
setting by introducing a similar notion of permutation invariance for $\mbY$.~
Specifically, for some set $\Yset \subset \mathbb{R}^N$, let
$\mathbb{Y}(\Yset) = \{ [\byy_1~\ldots~\byy_T]: \byy_t \in \Yset,~t=1,\ldots,T\}$
be the set of all $N\times T$ matrices whose columns are in $\Yset$, 
and let $\mbY(\Yset)$
be the set of all potential outcomes schedules whose matrices are all
elements of $\mathbb{Y}(\Yset)$. Thus, if
$\mY = [\bY(\zero), \bY(1), \bY(\be_2), \ldots, \bY(\be_T)] \in \mbY(\Yset)$ is one
such schedule of potential outcomes, then $\bY(\bz) \in \mathbb{Y}(\Yset)$ for all $\bz \in \Pulses$. 
%
%
\begin{definition}[Permutation-invariant schedule]\label{def:pi}
  Let $\Sset_N$ be the symmetric group on $N$
  elements. A set $\mbY$ of potential outcomes
  schedules is called permutation-invariant if $\mbY = \mbY(\Yset)$ for
  some $\Yset \subset \mathbb{R}^N$, such that $\Sset_N \cdot \Yset = \Yset$, 
  where $\Sset_N \cdot \Yset$ is the set in which every element of $\Yset$ has
  been permuted with every element of $\Sset_N$.
\end{definition}
The permutation invariance property in Definition~\ref{def:pi} captures a form of symmetry on 
the units' outcomes. For example, it implies that for any
$\bz \in \Pulses$ and any $t= 1, \ldots, T$, if it is possible that
$\byy_t(\bz) = \byy$, where $\byy_t(\bz) = (Y_{1t}(\bz), \ldots, Y_{Nt}(\bz))$, it should also be possible that
$\byy_t(\bz) = \pi \cdot \byy$, for any permutation $\pi \in \Sset_N$. This property
is not a probabilistic statement about the likelihood of $\byy$ or
$\pi \cdot \byy$, but a statement about the support of the potential outcomes, which is
a weaker assumption; see also \citep{wu1981minimax} for an interpretation of permutation
invariance in terms of robustness.

\subsection{Minimax optimal design}
\label{sec:minimax_design}

We can now state our first minimax theorem for pulse designs. The resulting design is
the solution to an integer optimization problem, which is generally hard to solve.
We therefore discuss a continuous relaxation in Proposition~\ref{prop:optimization}.

\begin{theorem}[Minimax pulse design]\label{th:minimax-design}
Let $\mbY$ be a bounded, permutation-invariant set of potential outcome schedules.~The minimax optimal pulse design, $\eta^\opt$, 
in Equation~\eqref{eq:minimax} is the completely randomized design that assigns $N_1^{\opt}$, $N_0^{\opt}$ and 
$\{N_{e_t}^{\opt}\}_{t=2}^T$ units to the assignments $\one$, $\zero$, and $\{\be_t\}_{t=2}^T$, respectively, where:
\begin{equation}\label{eq:IO}
	(N_0^\opt, N_1^\opt, \{N_{e_t}^{\opt}\}_{t=2}^T) = 
	\argmin_{\substack{N'_0, N'_1, \{N'_{e_t}\}_{t=2}^T \in \mathds{N}^+\\
				       N'_0 + N'_1 + \sum_{t=2}^T N'_{e_t} = N}}
	\bigg( \frac{T-1}{N'_1} + \frac{T-1}{N'_0} + 2 \sum_{t=2}^T \frac{1}{N'_{e_t}} \bigg).
\end{equation}
\end{theorem}

Two aspects of this result deserve special mention. First, the minimax design in
Theorem~\ref{th:minimax-design} is a completely randomized design. 
This result agrees with
the results of~\citet{wu1981minimax} and \citet{li1983minimaxity} who proved that
complete randomization is minimax optimal in the static, cross-sectional setting. Second, in contrast to the aforementioned works, the
minimax design in Theorem~\ref{th:minimax-design} is not balanced as it does not assign
the same number of units to each of the $T+1$ treatment arms in $\Pulses$.
This is made clearer by the following result.
\begin{proposition}\label{prop:optimization}
The relaxation of the integer optimization problem of Equation~\eqref{eq:IO},
\begin{equation}\label{eq:integer-optimization}
	(\tilde{N}_0, \tilde{N}_1, \{\tilde{N}_{e_t}\}_{t=2}^T) = 
	\argmin_{\substack{N'_0, N'_1, \{N_{e_t}^{\opt}\}_{t=2}^T \in \mathds{R}^+ \\
				       N'_0 + N'_1 + \sum_{t=2}^T N_{e_t} = N}}
	\bigg( \frac{T-1}{N'_1} + \frac{T-1}{N'_0} + 2 \sum_{t=2}^T \frac{1}{N'_{e_t}} \bigg),
\end{equation}
has the following solutions:
\begin{flalign*}
	\tilde{N}_\one = \tilde{N}_\zero  = \frac{N}{2 + \sqrt{2 (T-1)}};\quad \tilde{N}_{e_t} &= \sqrt{\frac{2}{T-1}} \frac{N}{2 + \sqrt{2(T-1)}},~t=2, \ldots, T.
\end{flalign*}
\end{proposition}
The solution of the relaxed problem exhibits a partial asymmetry with three notable features.
First, $\tilde{N}_1 = \tilde{N}_0$ while
$\tilde{N}_{e_t} = \tilde{N}_{e_t'}$ for all $1 \leq t,t' \leq T$, and so there is symmetry 
between always-control and always-treated units, but not across all units.
Second, $\tilde{N}_{\be_t} / \tilde{N}_\one \to 0$ as the time horizon $T$ grows, 
and so the number of units assigned to any pulse assignment is asymptotically negligible
compared to the always-treated (or always-control) assignments. Third,
$\tilde N_\one /  \sum_{t=2}^T \tilde N_{\be_t}  \to  0$, 
and so the pulse assignments taken together dominate the other treatment arms. 
To get intuition, let us consider the problem with $T=30$ and $N = 10000$. In this case, the minimax design assigns
$\tilde{N}_1 = \tilde{N}_0 \approx 1040$, and $\tilde{N}_{e_t} \approx 273$ for
$t=2, \ldots, 30$. In contrast, a balanced randomized design would assign
$322$ units to each treatment arm. 
%
Such difference between our design and standard balanced designs 
mainly stem from the definition of our loss function in Equation~\eqref{eq:loss},
which takes into account arbitrary-sized effects from past treatment history.

\begin{remark}
  Since the optimal design obtained in Theorem~\ref{th:minimax-design} is completely
  randomized, randomization-based analysis of the experiment is possible
  using standard Neymanian theory~\citep[Chapter 6]{imbens2015causal}. See proof in Appendix for details.
\end{remark}
  
\begin{remark}[Wedge designs]
\label{remark:wedges}
  As mentioned earlier, the results in this section (as well as in Section~\ref{sec:minimax-partial}) are
  stated in terms of pulse designs, but they can be extended to wedge designs since, by  Assumption~\ref{a:no-anticipating},
  \begin{equation*}
    Y_{it}(\bw_{t'}) = Y_{it}(\be_{t'}), \quad\text{for all},~t,t'~\text{such that}~t' \geq t.
  \end{equation*}
  That is, the potential outcomes of a unit under assignments $\bw_{t'}$ and $\be_{t'}$
  are identical for all periods prior to and including $t'$. In particular, the
  estimands, $\lambda_t$ and $\delta_t$, can be written in terms of wedge
  assignments instead of pulse assignments by substituting $\bw_t$ for $\be_t$ in
  Equation~\eqref{eq:composite}, and the rest of the analysis remains unchanged.
\end{remark}

\subsection{Minimax optimal design with augmented controls}
\label{sec:minimax-partial}

The design of Theorem~\ref{th:minimax-design} was obtained using plug-in
estimators for $\delta_t$ and $\lambda_t$. 
Here, we discuss the minimax design problem using a better estimator for $\delta_t$.
The key idea relies on 
Assumption~\ref{a:no-anticipating}, which implies that  $Y_{it}(\be_{t'}) = Y_{it}(\zero)$, for all $t < t'$.
In other words, under Assumption~\ref{a:no-anticipating}, at all times prior to $t'$ a pulse assignment
$\be_{t'}$ is indistinguishable from an always-control assignment in the sense that a unit assigned to $\be_{t'}$
behaves as if it had been assigned to $\zero$, for all $t < t'$. We can therefore use outcomes from 
units assigned to pulses for estimation of unknown control outcomes.

The new estimator that replaces the plug-in estimator, $\hat{\delta}_t$,
is defined as follows:
\begin{equation}\label{eq:gamma}
\hat{\gamma}_t = \frac{1}{N_{e_t}} \sum_i^N \iv(\bZ_i = \be_t)Y_{it}(\be_t) -
\frac{1}{N_t} \sum_i^N \iv(i \in \Cset^{(t)}) Y_{it}(\zero),
\end{equation}
where $N_t = |\Cset^{(t)}|$,  and $\Cset^{(t)} = \{i: \bZ_i = \zero,~\text{or}~\bZ_i = e_{t'},~t'>t\}$. Here, $\Cset^{(t)}$ is the new set of ``augmented control" units at $t$, i.e., 
the set comprised either of always-control units, or units assigned to pulse at a time $t'>t$.
The new loss function is now defined as
\begin{equation}\label{eq:loss2}
	L(\bZ, \mY) = \sum_{t=2}^T (\hatl_t - \lambda_t)^2 + \sum_{t=2}^T(\hat\gamma_t - \delta_t)^2,
\end{equation}
which only differs from Equation~\eqref{eq:risk} in using the new estimator, $\hat\gamma_t$. 
The updated minimax result is stated in the following theorem.
\begin{theorem} \label{th:partial-recycling}
Let $\mbY$ be a bounded, permutation-invariant set of potential outcome schedules.~The minimax optimal pulse design, $\eta^\opt$,
in Equation~\eqref{eq:minimax} using the new instantaneous effect estimator $\hat\gamma_t$
	in Equation~\eqref{eq:gamma}, is the completely randomized design that assigns $N_1^{\opt}$, $N_0^{\opt}$ and 
	$\{N_{e_t}^{\opt}\}_{t=2}^T$ units to the assignments $\one$, $\zero$, and $\{\be_t\}_{t=2}^T$, respectively, where:
	\begin{equation}\label{eq:integer-optimization-partial-recycle}
	(N_0^{\opt}, N_1^{\opt}, \{N_{e_t}^{\opt}\}_{t=2}^T) = 
	\argmin_{\substack{N'_0, N'_1, \{N'_{e_t}\}_{t=2}^T \in \mathds{N} \\
			N'_0 + N'_1 + \sum_{t=2}^T N'_{e_t} = N}}
	\bigg( \frac{T-1}{N'_1} + 2 \sum_{t=2}^T \frac{1}{N'_{e_t}} + \sum_{t=2}^T \frac{1}{N'_t} \bigg),
	\end{equation}
	and for each $t$,
        $N'_t = |\Cset^{(t)}| = N_0 + \sum_{t'=2}^T \iv(t' > t) N_{e_{t'}}$.
\end{theorem}
The main practical difference between Theorem~\ref{th:partial-recycling} and Theorem~\ref{th:minimax-design} is that the optimization problem of
Equation~\eqref{eq:integer-optimization-partial-recycle} is replacing that of
Equation~\eqref{eq:integer-optimization}. The following proposition derives the solutions
to a continuous relaxation of the problem.
\begin{proposition}\label{prop:solution-relaxation}
The following integer relaxation of the optimization problem in
Equation~\eqref{eq:integer-optimization-partial-recycle},
\begin{equation*}
	(N_0^{\opt}, N_1^{\opt}, \{N_{e_t}^{\opt}\}_{t=2}^T) = 
	\argmin_{\substack{N'_0, N'_1, \{N'_{e_t}\}_{t=2}^T \in \mathds{R} \\
			N'_0 + N'_1 + \sum_{t=2}^T N'_{e_t} = N}}
	\bigg( \frac{T-1}{N'_1} + 2 \sum_{t=2}^T \frac{1}{N'_{e_t}} + \sum_{t=2}^T \frac{1}{N'_{e_t}} \bigg),
\end{equation*}
has analytical solutions:
\begin{flalign*}
	N^{\opt}_{e_t} &= N^{\opt}_0 \sqrt{2} c_t, \quad t=2,\ldots, T,\\
	N^{\opt}_1 &= N - N^{\opt}_0 \bigg[ 1 + \sqrt{2} \sum_{t=2}^T c_t \bigg],\\
	N_0^{\opt} &= N \bigg[ 1 + (\sqrt{T-1} + \sqrt{2}) c_2 + \sqrt{2}\sum_{t=3}^T c_t  \bigg]^{-1},
\end{flalign*}
where $\{c_t\}_{t=2}^T$ are defined recursively by $c_T = 1$ and $c_t = \bigg[ \frac{1}{c_{t+1}^2} + \frac{1}{(1+ \sqrt{2} \sum_{t' > t} c_{t'})^2} \bigg]^{-1/2}$, for
all $t = 2, \ldots, T-1$.
\end{proposition}
The solution described by Proposition~\ref{prop:solution-relaxation} offers a sharp
contrast to the solution described by Proposition~\ref{prop:optimization}. 
Specifically, the solution of Proposition~\ref{prop:solution-relaxation}
does not exhibit the same form of partial symmetry as in Proposition~\ref{prop:optimization} 
because the 
new estimator, $\hat\gamma_t$, uses outcomes from pulse treatment as information 
about control potential outcomes. This reflects
the fundamental asymmetry of the loss function in Equation~\eqref{eq:loss2} in how it uses always-treated 
and always-control units. The effect will be illustrated more clearly in the simulations of Section~\ref{sec:simulations}, 
which will give more insight into this minimax design by comparing
it to the minimax design of Section~\ref{sec:minimax_design}, and to the standard, balanced
completely randomized design. 

\section{Extensions}
\label{sec:extensions}

\subsection{Weighted loss functions}
\label{sec:weighted-loss}

The loss functions considered in Equation~\eqref{eq:loss} and
Equation~\eqref{eq:loss2} put the same weight on the
instantaneous effects and the \dyneffs. This implicitly
assumes that both types of effects are of equal interest, which may not be true
in practice. 
More flexible loss functions could assign different weights to instantaneous and \dyneffs. In this section, we focus on extending the results of
Section~\ref{sec:minimax-partial}; analogous results for
Section~\ref{sec:minimax_design} can be derived.

Specifically, for some $\rho \in [0,1]$ consider the weighted loss function,
\begin{equation}\label{eq:wloss2}
  L(\bZ, \mY) = \rho \sum_{t=2}^T (\hatl_t - \lambda_t)^2 + (1-\rho)\sum_{t=2}^T(\hat\gamma_t - \delta_t)^2, 
\end{equation}
which generalizes Equation~\eqref{eq:loss2}. Thus, parameter $\rho$ controls the relative importance in estimating $\lambda_t$ or $\delta_t$.  The original loss function in Equation~\eqref{eq:loss} is a special case (up to a multiplicative constant)
with $\rho=1/2$. The following theorem extends Theorem~\ref{th:partial-recycling} 
to this new loss function, and derives the minimax optimal design as a function of $\rho$.
\begin{theorem}\label{th:minimax-weighted}
  Under the weighted loss function of Equation~\eqref{eq:wloss2}, the minimax
  optimal design, $\eta^{\opt}$, is still completely randomized, but with
  $N_1^{\opt}$, $N_0^{\opt}$ and $\{N_{e_t}^{\opt}\}_{t=2}^T$ solving:
  \begin{equation}\label{eq:WIOPR}
    (N_0^{\opt}, N_1^{\opt}, \{N_{e_t}^{\opt}\}_{t=2}^T) \quad = 
    \argmin_{\substack{N'_0, N'_1, \{N'_{e_t}\}_{t=2}^T \in \mathds{N} \\
        N'_0 + N'_1 + \sum_{t=2}^T N'_{e_t} = N}}
    \bigg( \rho \frac{T-1}{N'_1} +  \sum_{t=2}^T \frac{1}{N'_{e_t}} + (1-\rho)\sum_{t=2}^T \frac{1}{N'_t} \bigg),
  \end{equation}
  where for each $t$, $N'_t = |\Cset^{(t)}| = N_0 + \sum_{t'=2}^T \iv(t' > t) N_{e_{t'}}$.
\end{theorem}
In words, the minimax optimal design under a weighted loss function is still
completely randomized: what changes is the number of units assigned to each arm, depending on the parameter $\rho$.
The following relaxation of the optimization problem in Equation~\eqref{eq:WIOPR} provides more intuition.
\newpage

\begin{proposition}\label{prop:relaxation-weighted}
  The continuous relaxation of the optimization problem of Equation~\eqref{eq:WIOPR}
  has solutions, for $\rho < 1$:
  \begin{flalign*}
    N^{\opt}_{e_t} &= N^{\opt}_0 \ell c_t, \quad t=2,\ldots, T,\\
    N^{\opt}_1 &= N - N^{\opt}_0 \bigg[ 1 + \ell \sum_{t=2}^T c_t \bigg], \\
    N_0^{\opt} &= N \bigg[ 1 + \ell (1 + \sqrt{\rho(T-1)}) c_2,
    + \ell\sum_{t=3}^T c_t  \bigg]^{-1},
  \end{flalign*}
  where $\ell = (1-\rho)^{-1/2}$, $c_T = 1$ and
  $c_t = \bigg[ \frac{1}{c^2_{t+1}}
  + \frac{1}{(1 + \ell \sum_{t'>t}c_{t'})^2}\bigg]^{-1/2}$, for all $t=2, \ldots, T-1$. The
  case when $\rho = 1$ is obtained by symmetry with $\rho = 0$; see Appendix.
\end{proposition}

It is straightforward to see that when $\rho = 1/2$ we recover the results
of Proposition~\ref{prop:solution-relaxation}. Additional intuition can be obtained by
examining boundary values of $\rho$. Consider, for example, the case when $\rho = 1$, such that
the loss function only involves \dyneffs. Then, $N_0^{\opt} = 0$. This is reasonable: if we are only interested in estimating \dyneffs, then always-control units are not needed. Similarly, when the loss function only involves instantaneous effects ($\rho = 0$), the
minimax optimal design assigns no units to the always-treated arm ($N_1 = 0$).

\subsection{Recycling units when treatment effect attenuates}
\label{sec:recycling}

A fundamental premise of our approach so far is that we allow the effects of a treatment to persist indefinitely. In some cases, however, it may be reasonable to assume that the effect of the treatment 
wears off if a unit is untreated for a certain length of time, after being treated. We state this
assumption formally, and then derive the resulting minimax design.
\begin{assumption}[$k$-order carryovers]\label{a:k-recycling}
	For all $i=1, \ldots, N$ and all $t=1, \ldots, T$, 
	\begin{equation*}
		Y_{it}(\be_{t'}) = Y_{it}(\zero), \,\, \forall \,  0 <t' \leq t - k.
	\end{equation*}
\end{assumption}
Assumption~\ref{a:k-recycling} implies that after $k$ periods following a pulse treatment, $\be_{t'}$, the effects 
of the pulse
assignment  are indistinguishable from those of an always-control treatment.
Thus, a unit assigned to a pulse $\be_{t'}$ behaves as if it had been
assigned to $\zero$, for all time periods after (and including) $t' + k$.
Taken together, Assumption~\ref{a:no-anticipating} and Assumption~\ref{a:k-recycling}
suggest a new estimator of the instantaneous effect that generalizes the ``recycling estimator"  of Section~\ref{sec:minimax-partial}.

In particular, for some time period $t$, and $k\in\{1, \ldots, T\}$ with $t-k>0$, let $\Cset_k^{(t)}(\mZ) = \{i: \bZ_i = \zero \, \text{ or }
\bZ_i = e_{t'}, \,\, \text{ with }\, t' \leq t-k \, \text{ or } t' > t\}$ be the new set of ``augmented controls'',
comprised of units assigned either to always-control, to pulses before $t-k$, or to pulses after $t$. 
We replace the estimator of the instantaneous effect, $\hat{\delta}_t$, by the following estimator:
\begin{equation}\label{eq:recycling}
\hat{\beta}_t = \frac{1}{N_{e_t}} \sum_i^N \iv(\bZ_i = \be_t)Y_{it}(\be_t) -
\frac{1}{N_{t,k}} \sum_i^N \iv(i \in \Cset_k^{(t)}) Y_{it}(\zero),
\end{equation}
where $N_{t,k} = |\Cset_k^{(t)}|$. This new estimator is ``recycling" units that are assigned to pulses 
in order to estimate control outcomes~(term $\Cset^{(t)}_k$ in Equation~\eqref{eq:recycling}).
The loss and risk functions are as in
Section~\ref{sec:minimax}, but with $\hat{\beta}_t$ instead of $\hat{\delta}_t$
(or $\hat{\gamma}_t$). We can now state the minimax result under the new estimator of the instantaneous effect.

\begin{theorem} \label{th:recycling}
Let $\mbY$ be a bounded, permutation-invariant set of potential outcome schedules.~The minimax optimal pulse design, $\eta^\opt$,
in Equation~\eqref{eq:minimax} using the new instaneous effect estimator $\hat\beta_t$
	in Equation~\eqref{eq:recycling}
	is the completely randomized design that assigns $N_1^{\opt}$, $N_0^{\opt}$ and 
	$\{N_{e_t}^{\opt}\}_{t=2}^T$ units to the assigmnents $\one$, $\zero$, and $\{\be_t\}_{t=2}^T$, respectively, where:
	\begin{equation}\label{eq:integer-optimization-recycle}
	(N_0^{\opt}, N_1^{\opt}, \{N_{e_t}^{\opt}\}_{t=2}^T) = 
	\argmin_{\substack{N'_0, N'_1, \{N_{e_t}^{\opt}\}_{t=2}^T \in \mathds{N} \\
			N'_0 + N'_1 + \sum_{t=2}^T N'_{e_t} = N}}
	\bigg( \frac{T-1}{N'_1} + 2 \sum_{t=2}^T \frac{1}{N'_{e_t}} + \sum_{t=2}^T \frac{1}{N'_{t,k}} \bigg)
	\end{equation}
\end{theorem}
As before the optimal design is completely randomized but the  number of units assigned to each treatment arm 
differs from previous designs. The integer optimization problem
it involves (as well as its relaxation) is difficult to solve analytically. 
We plan to address this problem in future work.

\begin{remark}
The results of this section are specific to pulse
  designs in contrast to the results of Sections~\ref{sec:minimax_design},
  \ref{sec:minimax-partial} and \ref{sec:weighted-loss}, which also apply to wedge
  designs. Indeed, the fundamental idea of ``recycling" units is that if we wait long
  enough after a pulse, the units assigned to the pulse behave like control units. This does not apply when
  wedge designs are used, since units remain treated after the pulse.
\end{remark}


\section{Simulations}
\label{sec:simulations}

This section illustrates visually two aspects of our theory. In
Section~\ref{subsection:alloc}, we compute the number of units allocated to each
treatment arms for three designs: the balanced completely randomized design
(BCRD) that assigns the same number of units to all $T+1$ treatment arms, the
minimax design of Section~\ref{sec:minimax_design} and the augmented minimax
design of Section~\ref{sec:minimax-partial}. In Section~\ref{subsection:max-risk}
we quantify the reduction in maximum risk from using our minimax designs,
compared to the BCRD. This paper focuses exclusively on designs that minimize the
maximum risk: our theory says nothing about the expected risk. One cannot
reasonably expect minimax optimal designs to also be optimal for minimizing the
expected risk --- this is the price to pay for generality. Nevertheless, in
Appendix, we compare the expected risk of our minimax design to that of the
BCRD under two simple models, and show that our designs improve the expected
risk slightly.


\subsection{Treatment allocation}
\label{subsection:alloc}

We start by illustrating visually the optimal allocations obtained analytically in
Proposition~\ref{prop:optimization} and Proposition~\ref{prop:solution-relaxation}.
Figure~\ref{fig:allocation} shows the optimal allocation for the BCRD, the minimax
design and the augmented minimax design for $N=10000$ and for $T = 5, 10, 15$. 

As expected, the BCRD produces a fully symmetric allocation, shown as a horizontal blue
line. The standard minimax design of Section~\ref{sec:minimax_design} produces an
allocation that is only partially symmetric between two groups: $N_1=N_0$ and $N_{e_t}=N_{e_t'}$ for all $2\leq t,t'\leq T$; see Proposition~\ref{prop:optimization} and the subsequent discussion for details on such symmetry. We also confirm visually that the minimax
optimal design allocates more units than the BCRD to the always-treated arm
$(\bz = \one)$ and always-control arm $(\bz = \zero)$, but less in the pulses arms $(\bz = \be_t)$. On the other hand, the minimax design with augmented controls does not exhibit
a symmetry. It allocates very few units to the always-control arms: this is expected since
in this setting, some pulse units can be used as controls at each time $t$. We also
see that the number of units assigned to the pulse assignments $\{\be_t\}_{t=2}^T$ increases for larger values of $t$. This again is consistent with our intuition: as $t$
increases, the number of time periods for which a pulse can be used as control increases.
\begin{figure}[t!]
	\centering
	\includegraphics[width=1.1\textwidth]{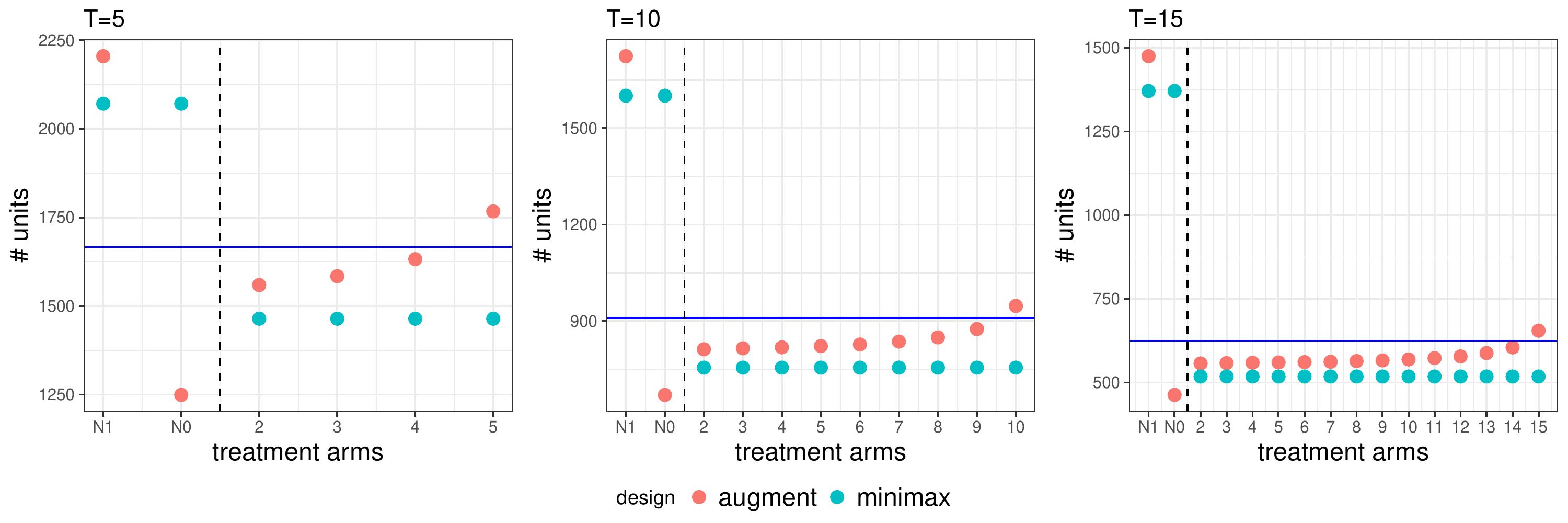}
	\vspace{-0.1in}\caption{ Unit allocations for BCRD, minimax design, and minimax design with augmented control. Blue horizontal line corresponds to BCRD, as a benchmark. In x-axis, $N_1$, $N_0$ correspond to always treated and always control, the numbers $2,\ldots, T$ correspond to the pulse designs $N_{e_2}, \ldots, N_{e_T}$.}
	\label{fig:allocation}  
\end{figure}

\subsection{Maximum risk}
\label{subsection:max-risk}

The risk under the augmented minimax design should be lower than
the risk under the minimax design, which in turn should be lower than the risk
under BCRD. In this section, we illustrate the magnitude of that reduction.
Figure~\ref{fig:maxrisk} plots the maximum risk for the minimax and augmented minimax
designs relative to the maximum risk of BCRD, for $N=1000$ and
$T = 10, 20, 30, 40, 50$.

We consider two settings, as shown in Figure~\ref{fig:maxrisk}. In the left panel, the estimator $\hat{\gamma}$ is used
for the risk under the augmented minimax but not for the risk under the other
designs; in the right panel, the estimator $\hat\gamma$ is used with all three
designs. In both settings, the maximum risk is always lower for the augmented minimax
design as predicted by theory. Compared to BCRD, the augmented minimax
design reduces the risk by up to $20\%$ compared to when $\hat{\lambda}$
is used for all designs. Given the trend further reduction would be expected for longer
time horizons. 

\begin{figure}[t!]
	\centering
	\includegraphics[width=0.9\textwidth]{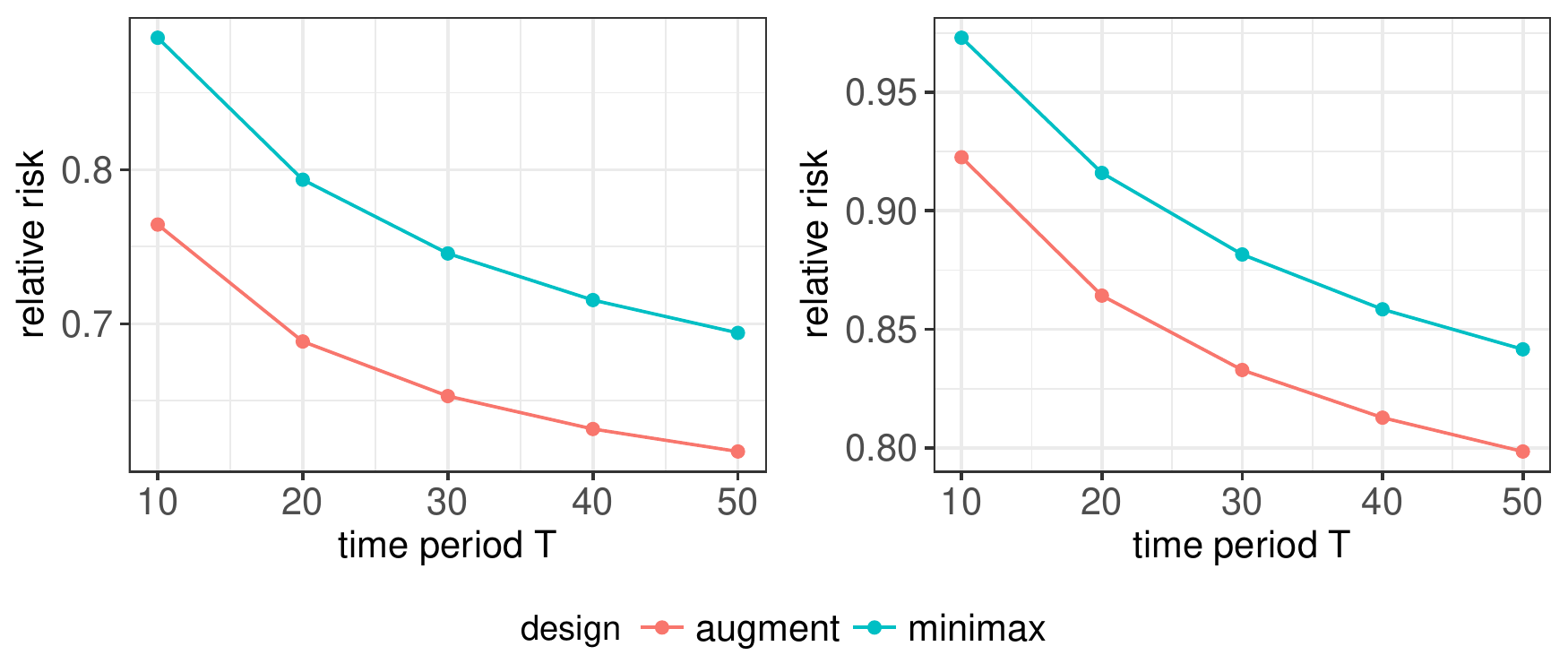}
	\vspace{-0.1in}\caption{
	Relative ratio of maximum risk using BCRD used as the baseline. Left: only augmented minimax design has augmented units. Right: both designs have augmented units.}
	\label{fig:maxrisk}  
\end{figure}

\section{Concluding remarks}
\label{sec:conclusion}

In this paper, we have constructed minimax optimal designs for estimating
the \dyn\ and instantaneous effects in temporal experiments. Our construction
uses the potential outcomes framework of causal inference, and is nonparametric
with mild assumptions on the support of potential outcomes. 

There are several open questions for future work.
First, from a technical perspective, it would be interesting to obtain 
analytical solutions for the recycling designs in Theorem~\ref{th:recycling}.
These should outperform the augmented minimax design because they wouldn't
throw away useful data. Second, in future work we would like to connect more
formally our estimands ($\lambda_t,\delta_t$ in Definition~\ref{def:estimands})
to estimands of long-term effects in user behavior experiments, which are popular
in digital
experimentation~\citep{hohnhold2015focusing, yan2019measuring, kohavi2009controlled}.
Third, we are looking at extensions of our approach, allowing covariate information
to be incorporated in the design.
Finally, our designs minimize the maximum risk, but they may be far from optimal
with respect to the expected risk. In future work, we plan on exploring
optimal designs for complex models of habituation.

\section{Acknowledgments}
We would like thank organizers and participants at the Google Market Workshop, and at the
LinkedIn Seminar for useful comments and feedback.

\bibliographystyle{chicago}
\bibliography{ref}

\newpage

\appendix

\singlespace
\small
\section{Proof of results in Section~\ref{sec:minimax_design}}

\subsection{Intermediate results and lemmas}

The proof of Theorem~\ref{th:minimax-design} has a number of intermediate steps, which
we will state as lemmas.

\smallskip

First, we define:
\begin{equation*}
	L_\lambda(\bZ, \mY) = \sum_{t=2}^T (\hat{\lambda}_t - \lambda_t)^2; \qquad
	L_\delta(\bZ, \mY) = \sum_{t=2}^T (\hat{\delta}_t - \delta_t)^2
\end{equation*}
so that we have $L(\bZ, \mY) = L_\lambda(\bZ, \mY) + L_\delta(\bZ, \mY)$. 
Next, we need to define the permutation actions on vectors and matrices.

\begin{definition}
  Let $\Sset_N$ be the symmetric group on $N$ elements, and $\pi \in \Sset_N$.
  If $\mathbf{v}$ is a vector of length $N$, then the action of $\pi$ on $\mathbf{v}$,
  denoted $\pi \cdot \mathbf{v}$, is the vector obtained by permuting the indices of 
  $\mathbf{v}$ according to $\pi$. Similarly, for any 
  population assignment $\bZ = [\bZ_1,\ldots, \bZ_N] \in \PulsesSupport \subset\{0, 1\}^{N\times T}$,
	%
  \begin{equation*}
    \pi \cdot	 \bZ \equiv [\bZ_{\pi^{-1}(1)}, \ldots, \bZ_{\pi^{-1}(N)}],
  \end{equation*}	
  where $\pi^{-1}(i)$ is the element $j$ that is mapped to $i$ through $\pi$.
\end{definition}
\begin{definition}
	For a potential outcomes schedule, $\mY = [\bY(\zero), \bY(\one), \bY(\be_2), \ldots, \bY(\be_T)]$,
	we define
	\begin{equation*}
          \pi \cdot \mY \equiv [\pi \cdot \bY(\zero), \pi \cdot \bY(\one),
          \pi \cdot \bY(\be_2), \ldots, \pi \cdot \bY(\be_T)]
	\end{equation*}
	where, as above, we have $\pi \cdot \bY(\bz) \equiv (\bY_{\pi^{-1}(1)}(\bz),
          \ldots, \bY_{\pi^{-1}(N)}(\bz))$,  for every assignment $\bz\in\Pulses$.
\end{definition}
We can now state and prove a sequence of lemmas. Throughout, we denote
$\bI = \{1, \ldots, N\}$. It is immediate to see that if $\pi \in \Sset_N$, then
$\pi \cdot \bI \equiv \{\pi(i), i \in \bI\} = \bI$ and thus $\pi^{-1} \cdot \bI = \bI$.
Our first lemma is to show that the loss function, $L$, is permutation-invariant in its arguments.

\begin{lemma}\label{lemma:loss}
For any $\pi \in \Sset_N$, and any assignment matrix $\bZ$ and schedule $\mY$,
	\begin{equation*}
		L(\pi\cdot \bZ, \pi \cdot \mY) = L(\bZ, \mY)
	\end{equation*}
\end{lemma}

\begin{proof}
  We show that $L_\delta(\pi\cdot \bZ, \pi \cdot \mY) = L_\delta(\bZ, \mY)$. The
  proof for $L_\lambda$ is identical, and so the proof for $L$ follows immediately. 
	
  \bigskip
  
  \noindent Recall that $\delta_t(\mY) = N^{-1} \sum_i Y_{it}(\be_t) -
  N^{-1} \sum_i Y_{it}(\zero)$. We have:
  \begin{flalign}
    \delta_t(\pi \cdot \mY) 
    &= N^{-1} \sum_{i\in \bI} (\pi \cdot \bY(\be_t))_{it} 
    - N^{-1} \sum_{i\in\bI} (\pi \cdot \bY(\zero))_{it} \notag\\
    &= N^{-1} \sum_{i\in\bI} Y_{\pi^{-1}(i),t}(\be_t) 
    - N^{-1} \sum_{i\in\bI} Y_{\pi^{-1}(i),t}(\zero) \notag \\	
    &= N^{-1} \sum_{j \in \pi^{-1} \cdot \bI} Y_{jt}(\be_t) 
    - N^{-1} \sum_{j\in \pi^{-1} \cdot \bI} Y_{jt}(\zero) \notag \\
    &= N^{-1} \sum_{j \in \bI} Y_{jt}(\be_t) - N^{-1}\sum_{j\in \bI} Y_{jt}(\zero)
    \notag \\
    &= \delta_t(\mY) \label{eq:estimand}.
  \end{flalign}
  
  \noindent We now turn to the estimator $\hat{\delta}_t$,
  \begin{flalign*}
    \hat{\delta}_t( \bZ, \mY) 
    &= N_{e_t}^{-1} \sum_{i\in \bI} \iv(\bZ_i = \be_t)Y_{it}(\be_t)
    -    N_0^{-1} \sum_{i\in \bI} \iv(\bZ_i = \zero_t)Y_{it}(\zero_t) \\
    & \equiv \hat{\delta}_t^{(\be_t)}( \bZ, \mY)  - \hat{\delta}_t^{(\zero)}( \bZ, \mY). 
  \end{flalign*}
With a similar argument as before:
	\begin{flalign*}
		 \hat{\delta}_t^{(\be_t)}( \pi \cdot \bZ, \pi \cdot \mY) 
		 &= N_{e_t}^{-1} \sum_{i \in \bI} \iv( (\pi \cdot \bZ)_i = \be_t) (\pi \cdot \bY(\be_t))_{it} \\
		 &= N_{e_t}^{-1} \sum_{i \in \bI} \iv( \bZ_{\pi^{-1}(i)} = \be_t) Y_{\pi^{-1}(i),t}(\be_t)\\
		 &= N_{e_t}^{-1} \sum_{j \in \pi^{-1} \cdot \bI} \iv(\bZ_j = \be_t) Y_{jt}(\be_t) \\
		 &= N_{e_t}^{-1} \sum_{j \in \bI} \iv(\bZ_j = \be_t) Y_{jt}(\be_t) \\
		 &= \hat{\delta}_t^{(\be_t)}(\bZ, \mY).
	\end{flalign*}
	The exact same reasoning leads to 
	$\hat{\delta}_t^{(\zero)}( \pi \cdot \bZ, \pi \cdot \mY) = \hat{\delta}_t^{(\zero)}(\bZ, \mY) $, and so:
	\begin{equation}\label{eq:estimator}
		\hat{\delta}_t(\pi \cdot \bZ, \pi \cdot \mY) = \hat{\delta}_t( \bZ, \mY).
	\end{equation}
	Putting together Equation~\eqref{eq:estimand} and Equation~\eqref{eq:estimator} we have:
	\begin{flalign*}
	L_\delta(\pi \cdot \bZ, \pi \cdot \mY) 
	&= \sum_{t=1}^T (\hat{\delta}_t(\pi\cdot \bZ, \pi \cdot \mY) - \delta_t(\pi \cdot \mY))^2 \\
	&= \sum_{t=1}^T (\hat{\delta}_t(\bZ, \mY) - \delta_t(\mY))^2 \\
	&= L_\delta(\bZ, \mY). 
	\end{flalign*}
\end{proof}


\begin{lemma}\label{lemma:symmetrization}
  For a pulse design $\eta \in \PulsesDesign$ and $\pi \in \Sset_N$, let
  $\eta_\pi$ be the design such that $\eta_\pi(\bZ) = \eta(\pi \cdot \bZ)$, and let
  $\tilde{\eta} = (N!)^{-1} \sum_{\pi \in \Sset_N} \eta_\pi$. Then, if $\mbY(\Yset)$ is permutation-invariant, we have:
  \begin{equation*}
    \max_{\mY \in \mbY(\Yset)}\{r(\tilde{\eta}; \mY)\} \leq \max_{\mY \in \mbY(\Yset)}\{r(\eta; \mY)\}.
  \end{equation*}	
\end{lemma}

\begin{proof}
We have:
\begin{flalign*}
	r(\tilde{\eta}; \mY) 
	&= \sum_{\bZ \in \PulsesSupport} \tilde{\eta}(\bZ) L(\bZ, \mY) \\
	&= \sum_{\bZ \in \PulsesSupport} \bigg[ (N!)^{-1} \sum_{\pi \in \Sset_N}\eta_\pi(\bZ)\bigg]L(\bZ, \mY) \\
	&= (N!)^{-1} \sum_{\pi \in \Sset_N} \sum_{\bZ \in \PulsesSupport} \eta_\pi(Z) L(\bZ, \mY).
\end{flalign*}
It follows that
\begin{flalign*}
\max_{\mY\in \mbY(\Yset)}{r(\tilde{\eta}; \mY)} 
&\leq (N!)^{-1} \sum_{\pi \in \Sset_N} \max_{\mY \in \mbY(\Yset)}\{\sum_{\bZ \in \PulsesSupport} \eta_\pi(\bZ)L(\bZ, \mY)\}\\
&= (N!)^{-1} \sum_{\pi \in \Sset_N} \max_{\mY \in \mbY(\Yset)}\{\sum_{\bZ \in \PulsesSupport} \eta(\pi \bZ)L(\bZ, \mY)\} \\
&=(N!)^{-1} \sum_{\pi \in \Sset_N} \max_{\mY \in \mbY(\Yset)}\{\sum_{\bZ \in \PulsesSupport} \eta(\pi\bZ)L(\pi\bZ, \pi\mY)\} \quad \text{[by Lemma~\ref{lemma:loss}]}\\
&=(N!)^{-1} \sum_{\pi \in \Sset_N} \max_{\mY \in \mbY(\Yset)}\{\sum_{\bZ' \in \pi \cdot \PulsesSupport} \eta(\bZ')L(\bZ', \pi\mY)\}\\
&=(N!)^{-1} \sum_{\pi \in \Sset_N} \max_{\mY \in \mbY(\Yset)}\{\sum_{\bZ' \in \PulsesSupport} \eta(\bZ')L(\bZ', \pi\mY)\},
\end{flalign*}
where the last equality follows from the fact that the support $\PulsesSupport$ is
permutation invariant. Now we use the fact that $\mbY(\Yset)$ is permutation invariant:
\begin{flalign*}
\max_{\mY \in \mbY(\Yset)}\{\sum_{\bZ' \in \PulsesSupport} \eta(\bZ')L(\bZ', \pi\mY)\}
&= \max_{\mY' \in \pi \cdot \mbY(\Yset)}\{\sum_{\bZ' \in \PulsesSupport} \eta(\bZ')L(\bZ', \mY') \} \\
&= \max_{\mY' \in \mbY(\Yset)}\{\sum_{\bZ' \in \PulsesSupport} \eta(\bZ')L(\bZ', \mY') \} \\
&= \max_{\mY' \in \mbY(\Yset)} r(\eta; \mY').
\end{flalign*}
Putting everything together, we obtain:
\begin{flalign*}
\max_{\mY\in \mbY(\Yset)}{r(\tilde{\eta}; \mY)} 
&\leq (N!)^{-1} \sum_{\pi \in \Sset_N} \max_{\mY \in \mbY(\Yset)}\{r(\eta; \mY)\}\\
&= \max_{\mY \in \mbY(\Yset)} \{r(\eta; \mY)\} \cdot (N!)^{-1} \sum_{\pi \in \Sset_N} 1 \\
&= \max_{\mY \in \mbY(\Yset)} r(\eta; \mY).
\end{flalign*}
\end{proof}

Next, we prove a representation lemma.
\begin{lemma}\label{lemma:crd}
  Let $\delta_{\bZ}$ be the design that assigns mass 1 at the assignment $\bZ$.
  Let $\eta \in \PulsesDesign$. Then:
  \begin{equation*}
    \tilde{\eta} = \sum_{\bZ \in \PulsesSupport} \eta(\bZ) \tilde{\delta}_{\bZ},
  \end{equation*}
  where $\tilde{\delta}_{\bZ} =  (N!)^{-1} \sum_{\pi \in \Sset_N} \xi_{\pi,\bZ}$  and $\xi_{\pi, \bZ}(\bZ') = \delta_{\bZ}(\pi \cdot \bZ')$.
\end{lemma}
\begin{proof}
From $\pi \cdot \PulsesSupport = \PulsesSupport$ it follows, for all $\pi\in\Sset_N$, that
\begin{flalign*}
 \eta_\pi  = \sum_{\bZ \in \PulsesSupport} \eta(\pi \cdot \bZ)
  \delta_{\bZ} 
  &= \sum_{\bZ' \in \pi \cdot \PulsesSupport} \eta(\bZ') \delta_{\pi^{-1}\cdot\bZ'}= \sum_{\bZ' \in \PulsesSupport} \eta(\bZ') \delta_{\pi^{-1}\cdot\bZ'}.
\end{flalign*}
By definition, $\delta_{\pi^{-1} \cdot \bZ} = \xi_{\pi, \bZ}$ since both function 
out mass 1 at the treatment $\pi^{-1} \cdot \bZ$. Then, from its definition in Lemma~\ref{lemma:symmetrization}:
%
%
%
\begin{flalign*}
	\tilde{\eta}  = (N!)^{-1} \sum_{\pi\in\Sset_N} \eta_\pi 
	&= \sum_{\bZ\in \PulsesSupport}\eta(\bZ) (N!)^{-1} \sum_{\pi \in \Sset_N} \xi_{\pi, \bZ} 
	= \sum_{\bZ \in \PulsesSupport} \eta(\bZ) \tilde{\delta}_{\bZ}.
\end{flalign*}
\end{proof}


\begin{lemma}\label{lemma:max-crd}
  Let $\eta \in \PulsesDesign$ and $\mbY(\Yset)$ be a permutation-invariant schedule of potential
  outcomes, where $\Yset$ is bounded. Then,
\begin{equation*}
	\max_{\mY \in \mbY(\Yset)} \{r(\tilde{\eta}; \mY)\} 
	= V^\ast \sum_{\bZ \in \PulsesSupport} \bigg( \frac{T-1}{N_1(\bZ)} + \frac{T-1}{N_0(\bZ)}
	+ 2 \sum_{t=2}^T \frac{1}{N_{e_t}(\bZ)}\bigg),
\end{equation*}
where
\begin{equation*}
	V^\ast = \max_{\bx \in \Yset} \frac{1}{N-1} \sum_{i=1}^N (\bx_i - \bar{\bx})^2 = O(1).
\end{equation*}
\end{lemma}

\begin{proof}
  We have:
  \begin{flalign*}
    r(\tilde{\eta}; \mY)
    & = \sum_{\bZ' \in \PulsesSupport} \tilde{\eta}(\bZ') L(\bZ', \mY) \\
    &= \sum_{\bZ' \in \PulsesSupport} \sum_{\bZ \in \PulsesSupport} \eta(\bZ) \tilde{\delta}_{\bZ}(\bZ') L(\bZ', \mY) 
    \quad [\text{from Lemma~\ref{lemma:crd}}]\\
    &= \sum_{\bZ \in \PulsesSupport} \bigg\{ 
    \sum_{\bZ' \in \PulsesSupport} \eta(\bZ) \tilde{\delta}_{\bZ}(\bZ') L(\bZ', \mY) \bigg\} \\
    &=\sum_{\bZ \in \PulsesSupport} \eta(\bZ) r(\tilde{\delta}_{\bZ}, \mY),
  \end{flalign*}
  where $\tilde{\delta}_\bZ$ is the CRD that assigns $N_0(\bZ)$ units to $\zero$,
  $N_1(\bZ)$ units to $\one$, 
  and $N_{e_t}(\bZ)$ units to pulse $\be_t$, for $t=2, \ldots, T$. The usual randomization sampling results hold \citep{imbens2015causal}:
  \begin{flalign*}
    r(\tilde{\delta}_{\bZ}; \mY) 
    &= \sum_{t=2}^T \bigg[\{\text{bias}_{\tilde{\delta}_{\bZ}}(\hat{\lambda}_t, \lambda_t)\}^2 
    + \{\text{bias}_{\tilde{\delta}_{\bZ}}(\hat{\delta}_t, \delta_t)\}^2\bigg] +
    \sum_{t=2}^T \bigg[V_{\tilde{\delta}_{\bZ}}(\hat{\lambda}_t)
    + V_{\tilde{\delta}_{\bZ}}(\hat{\delta}_t)\bigg] 
    &= \sum_{t=2}^T \bigg[V_{\tilde{\delta}_{\bZ}}(\hat{\lambda}_t)
    + V_{\tilde{\delta}_{\bZ}}(\hat{\delta}_t)\bigg], 
  \end{flalign*}
  where we defined:
  \begin{equation*}
	V_{\tilde{\delta}_{\bZ}}(\hat{\lambda}_t) =
	\frac{V^{(t)}_1}{N_1(\bZ)} + \frac{V^{(t)}_{e_t}}{N_{e_t}(\bZ)} - \frac{V^{(t)}_{1,e_t}}{N}; \qquad
	V_{\tilde{\delta}_{\bZ}}(\hat{\delta}_t) =
	\frac{V^{(t)}_0}{N_0(\bZ)} + \frac{V^{(t)}_{e_t}}{N_{e_t}(\bZ)} - \frac{V^{(t)}_{0,e_t}}{N},
\end{equation*}
and
\begin{flalign*}
	V^{(t)}_h &= (N-1)^{-1} \sum_{i=1}^N \{Y_{it}(h) - \overline{Y}_t(h)\}^2, \qquad h=\zero, \one, \be_t, 
	\quad t=2, \ldots, T\\
	V^{(t)}_{0,e_t} &= (N-1)^{-1} \sum_{i=1}^N \{ [Y_{it}(\be_t) - Y_{it}(\zero)] 
	- [\overline{Y}_t(\be_t) - \overline{Y}_t(\zero)]\}^2,\\
	V^{(t)}_{1,e_t} &= (N-1)^{-1} \sum_{i=1}^N \{ [Y_{it}(\one) - Y_{it}(\be_t)] 
	- [\overline{Y}_t(\one) - \overline{Y}_t(\be_t)]\}^2,
\end{flalign*}
where $\bar Y_t$ indicates averaging over all units.
Occasionally, we will write $V_h^{(t)}(\mY), V_{0, e_t}(\mY), V_{1, e_t}(\mY)$ to emphasize that these are functions of the potential outcomes 
schedule, $\mY$.
Note also that the bias terms are zero because the estimators, $\hat\lambda_t, \hat\delta_t$ are unbiased.
Putting everything together, we obtain:
\begin{equation}\label{eqR}
r(\tilde{\eta}; \mY) = \sum_{t=2}^T \bigg[ V^{(t)}_1 \sum_{\bZ\in \PulsesSupport} \frac{\eta(\bZ)}{N_1(\bZ)}
+ V^{(t)}_0 \sum_{\bZ\in \PulsesSupport} \frac{\eta(\bZ)}{N_0(\bZ)} 
+ 2V^{(t)}_{e_t} \sum_{\bZ\in \PulsesSupport} \frac{\eta(\bZ)}{N_{\be_t}(\bZ)}
- \frac{V_{0,e_t}^{(t)}}{N} - \frac{V_{1,e_t}^{(t)}}{N}
\bigg].
\end{equation}
Now is the key part of the argument. 
First, 
$V_h^{(t)}$ is a function only of $\by_t(h) = (Y_{1t}(h), \ldots, Y_{Nt}(h))$, 
say, $V_h^{(t)}= V(\by_t(h))$.
Furthermore, $\by_t(\bh)$ takes values from $\Yset$.
Therefore,
\begin{equation*}
 \argmax_{\by_t(\mathbf{h}) \in \Yset} V(\by_t(\mathbf{h}))
  = \text{constant} \equiv \Yset^{\opt},~\text{for all}~h,t;
\end{equation*}
here, we allow $\argmax$ to return a set. In particular, $\Yset^{\opt}$ is a set because there might be many vectors that maximize $V$; the set can be a singleton, but it is not empty.
Now, for every element $\by\in\Yset^{\opt}$, it holds
\begin{equation}\label{eqV}
\max_{\mY \in \mbY(\Yset)} V_h^{(t)} = V(\by),
\end{equation}
since $\mbY(\Yset)$ contains all potential outcomes 
schedules, such that the columns of every matrix in every schedule 
are from $\Yset$.
Take any vector $\by \in \Yset^{\opt}$. Define $\bY^{\opt}$ as the $N\times T$ matrix where each column is equal to $\by$, and define
$\mY^{\opt} = [\bY^{\opt}, \ldots, \bY^{\opt}]$ the potential outcomes schedule
containing $T+1$ copies of $\bY^{\opt}$. Then, for all $h = \zero, \one, \be_t$, and 
all $t=2,\ldots, T$, 
\begin{equation}\label{eqV2} 
	\max_{\mY \in \mbY(\Yset)} V_h^{(t)}
	= V_h^{(t)}(\mY^{\opt}) = \frac{1}{N-1} \sum_{i=1}^N(\by_i - \overline{\by})^2 \equiv V^\ast = O(1),
\end{equation}
since the potential outcomes are bounded. It then follows that:

%

%
\begin{equation*}
\mY^{\opt} \in \argmax_{\mY \in \mbY(\Yset)} \bigg\{\sum_{t=2}^T 
\bigg[ V^{(t)}_1 \sum_{\bZ\in \PulsesSupport} \frac{\eta(\bZ)}{N_1(\bZ)}
+ V^{(t)}_0 \sum_{\bZ\in \PulsesSupport} \frac{\eta(\bZ)}{N_1(\bZ)} 
+ 2V^{(t)}_{e_t} \sum_{\bZ\in \PulsesSupport} \frac{\eta(\bZ)}{N_1(\bZ)} \bigg]\bigg\},
\end{equation*}
since we can maximize $V_1^{(t)}, V_0^{(t)}, V_{e_t}^{(t)}$ separately
for every $t$ by construction of the $\mY^{\opt}$ and the argument 
in Equation~\eqref{eqV}.

Now, we need to turn our attention to the negative terms in Equation~\eqref{eqR}. 
First, for all $t$, 
$$V_{0,e_t}^{(t)}(\mY^{\opt}) = V_{1,e_t}^{(t)}(\mY^{\opt}) = 0,
$$
by definition of $\mY^{\opt}$. But we also have:
\begin{equation*}
	\max_{\mY \in \mbY(\Yset)} \{-V^{(t)}_{0,e_t}(\mY) \} 
	 = \max_{\mY \in \mbY(\Yset)} \{-V^{(t)}_{1,e_t}(\mY) \} = 0,
\end{equation*}
because the $V$-functions are non-negative.
It follows that
\begin{equation}\label{eqY2}
	\mY^{\opt} \in \argmax r(\tilde{\eta}_{\bZ}; \mY).
\end{equation}
We conclude that
\begin{flalign*}
	\max_{\mY \in \mbY(\Yset)} r(\tilde{\eta}_{\bZ}; \mY) 
	&=\sum_{t=2}^T \bigg[V^\ast \sum_{\bZ \in \PulsesSupport} \eta(\bZ) 
	\bigg\{\frac{1}{N_1(\bZ)} + \frac{1}{N_0(\bZ)} + \frac{2}{N_{e_t}(\bZ)}\bigg\}\bigg] \\
	&= V^\ast \sum_{\bZ \in \PulsesSupport} \eta(\bZ) 
	\bigg\{\frac{T-1}{N_1(\bZ)} + \frac{T-1}{N_0(\bZ)} + 2\sum_{t=2}^T\frac{1}{N_{e_t}(\bZ)}\bigg\},
\end{flalign*}
where we combined Equation~\eqref{eqV2} and Equation~\eqref{eqY2}.
\end{proof}

\subsection{Proof of Theorem~\ref{th:minimax-design}}

\begin{proof}
  If $\eta$ is minimax optimal, then by Lemma~\ref{lemma:symmetrization},
  $\tilde{\eta}$ is also minimax; our strategy, therefore, will be to find the design
  of the form $\tilde{\eta}$ that achieves the minimax risk. For a design
  $\tilde{\eta}$, we have, by Lemma~\ref{lemma:max-crd}, 
\begin{flalign*}
\max_{\mY \in \mbY(\Yset)} r(\tilde{\eta}_{\bZ}; \mY) 
&= V^\ast \sum_{\bZ \in \PulsesSupport} \eta(\bZ) 
\bigg\{\frac{T-1}{N_1(\bZ)} + \frac{T-1}{N_0(\bZ)} + 2\sum_{t=2}^T\frac{1}{N_{e_t}(\bZ)}\bigg\}\\
&= \sum_{\bZ \in \PulsesSupport} \eta(\bZ) 
\bigg\{\frac{T-1}{N_1(\bZ)} + \frac{T-1}{N_0(\bZ)} + 2\sum_{t=2}^T\frac{1}{N_{e_t}(\bZ)}\bigg\},
\end{flalign*}
since $V^\ast$ does not depend on $\bZ$. It follows that
$\min_\eta \max_{\mY \in \mbY(\Yset)} r(\tilde{\eta}_{\bZ}; \mY)$ is attained by the
design $\eta$ that satisfies:
\begin{equation}\label{eq:eta-opt}
	\eta(\bZ) > 0 \quad \Rightarrow \quad (N_1(\bZ), N_0(\bZ), \{N_{e_t}(\bZ)\}) 
	= \argmin_{N'_1, N'_0, \{N'_{e_t}\}} \bigg\{\frac{T-1}{N'_1} 
	+ \frac{T-1}{N'_0} + 2\sum_{t=2}^T\frac{1}{N'_{e_t}}\bigg\}.
\end{equation}
Let $\eta^{\opt}$ be a design satisfying Equation~\eqref{eq:eta-opt}, and consider
$\tilde{\eta}^{\opt}$. For any $\bZ$ such that $\eta^{\opt}(\bZ) > 0$, we have:
\begin{flalign*}
	\tilde{\eta}^{\opt}_\pi(\bZ) 
	&= \tilde{\eta}^{\opt}(\pi \cdot \bZ) \\
	&= (N!)^{-1}\sum_{\pi' \in S_N} \eta^{\opt}_{\pi'}(\pi \cdot \bZ) \\
	&= (N!)^{-1}\sum_{\pi' \in S_N} \eta^{\opt}( \pi'\pi \cdot \bZ) \\
	&= (N!)^{-1}\sum_{\pi'' \in S_N} \eta^{\opt}( \pi'' \cdot \bZ)\\
	&= \tilde{\eta}^{\opt}(\bZ),
\end{flalign*}
for any permutation $\pi \in \Sset_N$; that is,
$\tilde{\eta}^{\opt}(\pi \cdot \bZ) = \tilde{\eta}^{\opt}(\bZ) = c$, for any permutation
$\pi \in \Sset_N$. But the permutations of $\bZ$ are the assignments such that:
\begin{equation}\label{eq:assgt-min}
(N_1(\bZ), N_0(\bZ), \{N_{e_t}(\bZ)\}) 
	= \argmin_{N'_1, N'_0, \{N'_{e_t}\}} \bigg\{\frac{T-1}{N'_1} 
	+ \frac{T-1}{N'_0} + 2\sum_{t=2}^T\frac{1}{N'_{e_t}}\bigg\}.
\end{equation}
It is easy to verify that $\tilde{\eta}^{\opt}$ assigns mass zero to the assignments
$\bZ$ that do not satisfy Equation~\eqref{eq:assgt-min}, and so in conclusion,
$\tilde{\eta}^{\opt}$ is the completely randomized design with
$N_1^{\opt}$, $N_0^{\opt}$ and $\{N_{e_t}^{\opt}\}$ satisfying
Equation~\eqref{eq:assgt-min}.

\end{proof}

\subsection{Proof of Proposition~\ref{prop:optimization}}

\begin{proof}
  The Lagrangian function corresponding to the objective function is:
  \begin{align*}
    \mathcal{L}(N'_0, N'_1, \{N'_{e_t}\}_{t=2}^T,\lambda)=\bigg( \frac{T-1}{N'_0} +\frac{T-1}{N'_1} + 2 \sum_{t=2}^T \frac{1}{N'_{et}} \bigg) +\lambda\bigg(N'_0 + N'_1 + \sum_{t=2}^T N_{e_t} - N\bigg)
  \end{align*}
  Setting the partial derivatives of the Lagrangian to zero yields the following
  system of equations:
  \begin{flalign*}
    \frac{\partial \mathcal{L}}{\partial N'_0} &= -\frac{T-1}{{N'}_0^2}+\lambda=0 \\
    \frac{\partial \mathcal{L}}{\partial N'_1} &= -\frac{T-1}{{N'}_1^2}+\lambda=0 \\
    \frac{\partial \mathcal{L}}{\partial N'_{e_t}}
    &= -2\frac{1}{(N'_{e_t})^2}+\lambda=0, \quad t=2, \ldots T
  \end{flalign*}
  whose solution we write as a function of $\lambda$:
  \begin{align*}
    N'_0 = \sqrt{\frac{T-1}{\lambda}}, \quad N'_1 = \sqrt{\frac{T-1}{\lambda}}, \quad
    N'_{e_t} = \sqrt{\frac{2}{\lambda}}
  \end{align*}
  But the total number of units assigned must add up to $N$, which allows us to solve
  for $\lambda$:
  \begin{align*}
    N'_0 + N'_1 + \sum_{t=2}^T N_{e_t}=
    2\sqrt{\frac{T-1}{\lambda}}+(T-1)\sqrt{\frac{2}{\lambda}}= N
    \implies \sqrt{\frac{1}{\lambda}}=\frac{N}{2\sqrt{(T-1)}+\sqrt{2}(T-1)}
  \end{align*}
  leading to the following solution:
  \begin{flalign*}
    \tilde{N}_1 &= \frac{N}{2 + \sqrt{2 (T-1)}}, \\
    \tilde{N}_0 &=\frac{N}{2 + \sqrt{2 (T-1)}}, \\
    \tilde{N}_{e_t} &= \sqrt{\frac{2}{T-1}} \frac{N}{2 + \sqrt{2(T-1)}}, \qquad t=2, \ldots, T.
  \end{flalign*}
  This completes the proof. 
\end{proof}

\section{Proof of results in Section~\ref{sec:minimax-partial}}

\subsection{Proof of Theorem~\ref{th:partial-recycling}}

The proof of Theorem~\ref{th:partial-recycling} follows exactly the same lines as
that of Theorem~\ref{th:recycling} which we prove below in
Appendix~\ref{appendix:recycling}.

\subsection{Proof of Proposition~\ref{prop:solution-relaxation}}

\begin{proof}
	Consider the objective function:
	\begin{equation*}
		\rho(N_0, N_{e_2}, \ldots, N_{e_T}) = \frac{T-1}{N - N_0 -\sum_{t'=2}^T N_{e_{t'}}} 
		+ 2 \sum_{t'=2}^T \frac{1}{N_{e_{t'}}} + \sum_{t'=2}^T \frac{1}{N'_t}
	\end{equation*}
	then notice that:
\begin{flalign*}
\frac{d\rho}{dN_{e_t}} &= \frac{T-1}{[N - N_0 -\sum_{t'=2}^T N_{e_{t'}}]^2} 
- \frac{2}{(N_{e_t})^2} - \sum_{t'=2}^T \iv(t' < t) \frac{1}{[N_0 + \sum_{t''=2}^T \iv(t'' > t') N_{e_{t''}}]^2} \\
\frac{d\rho}{dN_{e_{t-1}}} &= \frac{T-1}{[N - N_0 -\sum_{t'=2}^T N_{e_{t'}}]^2} 
- \frac{2}{(N_{e_{t-1}})^2} - \sum_{t'=2}^T \iv(t' < t-1) \frac{1}{[N_0 + \sum_{t''=2}^T \iv(t'' > t') N_{e_{t''}}]^2}
\end{flalign*}
Make $\frac{d\rho}{dN_{e_t}}=0$ and $\frac{d\rho}{dN_{e_{t-1}}}=0$, and the terms cancel between them by subtraction as follows:
\begin{flalign*}
\frac{2}{(N_{e_t})^2} = \frac{2}{(N_{e_{t-1}})^2} - \frac{1}{[N_0 + \sum_{t''=2}^T \iv(t'' > t-1) N_{e_{t''}}]^2}
\end{flalign*}
Therefore, we can write the same equation for $t=3,\ldots, T$
\begin{flalign*}
\frac{2}{(N_{e_T})^2} &= \frac{2}{(N_{e_{T-1}})^2} - \frac{1}{[N_0 + \sum_{t''=2}^T \iv(t'' > T-1) N_{e_{t''}}]^2} \\
\frac{2}{(N_{e_{T-1}})^2} &= \frac{2}{(N_{e_{T-2}})^2} - \frac{1}{[N_0 + \sum_{t''=2}^T \iv(t'' > T-2) N_{e_{t''}}]^2} \\
&\vdots \\
\frac{2}{(N_{e_3})^2} &= \frac{2}{(N_{e_{2}})^2} - \frac{1}{[N_0 + \sum_{t''=2}^T \iv(t'' > 2) N_{e_{t''}}]^2} 
\end{flalign*}
Add all these equations on the left and right sides, we have
\begin{flalign}\label{eq: iter sum}
\frac{2}{(N_{e_T})^2} = \frac{2}{(N_{e_{2}})^2} - \sum_{t'=2}^{T-1} \frac{1}{[N_0 + \sum_{t''=2}^T \iv(t'' > t') N_{e_{t''}}]^2}
\end{flalign}
Also, we have
\begin{flalign}
\frac{d\rho}{dN_0} &= \frac{T-1}{[N - N_0 -\sum_{t'=2}^T N_{e_{t'}}]^2} 
- \sum_{t'=2}^T \frac{1}{[N_0 + \sum_{t''=2}^T \iv(t'' > t') N_{e_{t''}}]^2} =0 \label{eq: part N0}\\
\frac{d\rho}{dN_{e_2}} &= \frac{T-1}{[N - N_0 -\sum_{t'=2}^T N_{e_{t'}}]^2} 
- \frac{2}{(N_{e_2})^2} =0 \label{eq: part N2}
\end{flalign}
Plug these two equations to Equation~\ref{eq: iter sum}, we have
\begin{flalign*}
\frac{2}{(N_{e_T})^2} = \frac{1}{(N_0)^2} \implies N_{e_T} = N_0 \sqrt{2}c_T
\end{flalign*}
where $c_T=1$. Plug this to the $\frac{2}{(N_{e_T})^2} = \frac{2}{(N_{e_{T-1}})^2} - \frac{1}{[N_0 + N_{e_T}]^2}$, we have
\begin{flalign*}
N_{e_{T-1}} = N_0 \sqrt{2} c_{T-1},
\end{flalign*}
where $c_{T-1} = \bigg[\frac{1}{c_T^2} + \frac{1}{(1+\sqrt{2}c_T)^2} \bigg]^{-1/2}$. So, do the same operation for other equations, we have
\begin{flalign*}
N_{e_t} = N_0 \sqrt{2} c_t, \quad t=2,\ldots,T
\end{flalign*}
where 
\begin{flalign*}
c_t &= \bigg[ \frac{1}{c_{t+1}^2} + \frac{1}{(1+ \sqrt{2} \sum_{t' > t} c_{t'})^2} \bigg]^{-1/2} 
\quad t = 2, \ldots, T-1 \\
c_T &= 1
\end{flalign*}
Plug these $\{N_{e_t}\}$ to Equation~\ref{eq: part N2}, we have
\begin{flalign*}
\frac{T-1}{[N - N_0 -\sum_{t'=2}^T N_{e_{t'}}]^2} 
- \frac{2}{(N_{e_2})^2} = 0 \implies \frac{2}{2c_2^2N_0^2} = \frac{T-1}{[N - N_0 - \sqrt{2}_{t=2}^Tc_t N_0 ]^2}
\end{flalign*}
Solve this equation, we get 
\begin{flalign*}
N_0 = \bigg[ 1 + (\sqrt{T-1} + \sqrt{2}) c_2 + \sqrt{2}\sum_{t=3}^T c_t  \bigg]^{-1}
\end{flalign*}
Finally, we have
\begin{flalign*}
N_1 &= N - N_0 \bigg[ 1 + \sqrt{2} \sum_{t=2}^T c_t \bigg]
\end{flalign*}
\end{proof}

\section{Proof of results in Section~\ref{sec:weighted-loss}}

\subsection{Proof of Theorem~\ref{th:minimax-weighted}}

Lemma~\ref{lemma:loss} through Lemma~\ref{lemma:crd} hold trivially with the
weighted loss function. Adapting Lemma~\ref{lemma:max-crd} requires more care.
\begin{lemma}[Analog to Lemma~\ref{lemma:max-crd}]\label{lemma:max-crd-weighted}
  Let $\eta \in \PulsesDesign$ and $\mbY(\Yset)$ where $\Yset$ is bounded and
  permutation-invariant, then:
  \begin{equation*}
    \max_{\mY \in \mbY(\Yset)} V^\ast \sum_{\bZ \in \PulsesSupport}
    \bigg[ \rho \frac{T-1}{N_1(\bZ)}
    + (1-\rho) \sum_{t=2}^T\frac{1}{N_t(\bZ)}
    + \sum_{t=2}^T \frac{1}{N_{e_t}(\bZ)}
    \bigg]
  \end{equation*}
  where $V^\ast$ is as in Lemma~\ref{lemma:max-crd}.
\end{lemma}
\begin{proof}[Proof of Lemma~\ref{lemma:max-crd-weighted}]
  The proof requires to make some of same modifications made in the proof of
  Lemma~\ref{lemma:max-crd-alt}. Here, we give a high-level view of some of the changes
  in the proof. We have:
  \begin{flalign*}
    r(\tilde{\delta}_{\bZ}; \mY)
    &= \sum_{t=2}^T \bigg( \rho \text{Bias}(\hat{\lambda}_t)
    + (1-\rho) \text{Bias}(\hat{\delta}_t)\bigg)
    + \sum_{t=2}^T \bigg( \rho \text{Var}_{\hat{\delta}_\bZ}(\hat{\lambda}_t)
    + (1-\rho) \text{Var}_{\hat{\delta}_\bZ}(\hat{\delta}_t)\bigg)\\
    &= \sum_{t=2}^T \bigg( \rho \text{Var}_{\hat{\delta}_\bZ}(\hat{\lambda}_t)
    + (1-\rho) \text{Var}_{\hat{\delta}_\bZ}(\hat{\delta}_t)\bigg)
  \end{flalign*}
  since the bias is zero under complete randomization. We therefore have (as in the
  proof of Lemma~\ref{lemma:max-crd}):
  \begin{equation*}
    r(\tilde{\eta}; \mY)
    = \sum_{t=2}^T \bigg[
    V_1^{(t)} \sum_{\bZ \in \PulsesSupport} \rho \frac{\eta(\bZ)}{N_1(\bZ)}
    + V_0^{(t)} \sum_{\bZ \in \PulsesSupport} (1-\rho) \frac{\eta(\bZ)}{N_t(\bZ)}
    + V_{e_t}^{(t)} \sum_{\bZ \in \PulsesSupport} (\rho+1-\rho)
    \frac{\eta(Z)}{N_{e_t}(\bZ)}
    - \frac{V_{0,e_t}}{N} - \frac{V_{1,e_t}}{N}\bigg]
  \end{equation*}
  and it follows that:
  \begin{equation*}
    \max_{\mY \in \mbY(\Yset)} V^\ast \sum_{\bZ \in \PulsesSupport}
    \bigg[ \rho \frac{T-1}{N_1(\bZ)}
    + (1-\rho) \sum_{t=2}^T\frac{1}{N_t(\bZ)}
    + \sum_{t=2}^T \frac{1}{N_{e_t}(\bZ)}
    \bigg]
  \end{equation*}
  which concludes the proof.
\end{proof}

\begin{proof}[Proof of Theorem~\ref{th:minimax-weighted}]
  The proof follows exactly the lines of the proof of Theorem~\ref{th:minimax-design}
  with the modified lemmas.
\end{proof}

\subsection{Proof of Proposition~\ref{prop:relaxation-weighted}}

The proof is similar to that of Proposition~\ref{prop:solution-relaxation}, so we
focus on the parts that change:
\begin{proof}[Proof of Proposition~\ref{prop:relaxation-weighted}]
  Let
  \begin{equation*}
    \phi(N_0, N_{e_2}, \ldots, N_{e_T})
    = \rho \frac{T-1}{N - N_0 - \sum_{t'=2}^T N_{e_{t'}}}
    + \sum_{t'=2}^T \frac{1}{N_{e_{t'}}}
    + (1-\rho) \sum_{t'=2}^T \frac{1}{N_0 + \sum_{t''=2}^T
      \iv\{t'' > t'\} N_{e_{t''}}}
  \end{equation*}
  where we use $\phi$ instead of $\rho$ for the objective function, since $\rho$ is
  already used to denote the weight in this section. Now notice that:
  \begin{equation*}
    \frac{d \phi}{d N_{e_t}} =
    \rho \frac{T-1}{(N - N_0 - \sum_{t'=2}N_{e_{t'}})^2}
    - \frac{1}{N_{e_t}^2}
    - \sum_{t'=2}^T \iv\{t' < t\} (1-\rho) \frac{1}{(N_0 + \sum_{t''=2}^T
      \iv\{t'' > t'\} N_{e_{t''}})^2}
  \end{equation*}
  and
  \begin{equation*}
    \frac{d \phi}{d N_{e_{t-1}}} =
    \rho \frac{T-1}{(N - N_0 - \sum_{t'=2}N_{e_{t'}})^2}
    - \frac{1}{N_{e_{t-1}}^2}
    - \sum_{t'=2}^T \iv\{t' < t-1\} (1-\rho) \frac{1}{(N_0 + \sum_{t''=2}^T
      \iv\{t'' > t'\} N_{e_{t''}})^2}
  \end{equation*}
  and so:
  \begin{flalign*}
    \frac{d \phi}{dN_{e_t}} = \frac{d \phi}{dN_{e_{t-1}}} = 0
    &\quad \Leftrightarrow \quad
    \frac{d \phi}{dN_{e_t}} - \frac{d \phi}{dN_{e_{t-1}}} = 0 \\
    &\quad \Leftrightarrow \quad
     \frac{1}{N_{e_{t-1}}^2} 
    - \frac{1}{N_{e_t}^2}
    - (1-\rho) \frac{1}{(N_0 + \sum_{t''=2}^T \iv\{t'' > t-1\} N_{e_{t''}})^2} = 0\\
    &\quad \Leftrightarrow \quad
     \frac{1}{N_{e_t}^2} = 
    \frac{1}{N_{e_{t-1}}^2} 
    - (1-\rho) \frac{1}{(N_0 + \sum_{t''=2}^T \iv\{t'' > t-1\} N_{e_{t''}})^2}.
  \end{flalign*}
  By telescoping the sums, we have
  \begin{equation*}
    \frac{1}{N_{e_T}^2} = \frac{1}{N_{e_2}^2}
    - (1-\rho) \sum_{t'=2}^{T-1}
    \frac{1}{(N_0 + \sum_{t''=2}^T \iv\{t''>t'\}N_{e_{t''}})^2}
  \end{equation*}
  In addition, we have:
  \begin{flalign}
    \frac{d\rho}{dN_0} &= \rho \frac{T-1}{[N - N_0 -\sum_{t'=2}^T N_{e_{t'}}]^2} 
    - (1-\rho) \sum_{t'=2}^T
    \frac{1}{[N_0 + \sum_{t''=2}^T \iv(t'' > t') N_{e_{t''}}]^2} \label{eq:N0}\\
    \frac{d\rho}{dN_{e_2}} &= \rho \frac{T-1}{[N - N_0 -\sum_{t'=2}^T N_{e_{t'}}]^2} 
    - \frac{1}{N_{e_2}^2}
  \end{flalign}
  Combining with the previous equation, as in the proof of
  Proposition~\ref{prop:solution-relaxation} yields:
  \begin{equation*}
    \frac{1}{N_{e_T}^2} = (1-\rho) \frac{1}{N_0^2}
  \end{equation*}
  and so $N_{e_T} = N_0 \ell c_T$ where $\ell = (1-\rho)^{-1/2}$ and $c_T = 1$.
  Plugging this into the recursive definition of $N_{e_t}$, starting from
  $N_{e_T}$ we get
  \begin{flalign*}
    \frac{1}{N_{e_{T-1}}^2}
    &= \frac{1}{N_{e_T}^2} + \frac{(1-\rho)}{(N_0 + N_{e_T})^2}\\
    &= \frac{1}{N_0^2}
    \bigg[ \frac{1}{(\ell c_T)^2} + \frac{(1-\rho)}{(1+\ell c_T)} \bigg]
  \end{flalign*}
  which implies:
  \begin{flalign*}
    N_{e_{T-1}}
    &= N_0 (1-\rho)^{-1/2} \bigg[ \frac{1}{(1-\rho)\ell^2 c_T^2}
    + \frac{1}{(1 + \ell c_T)^2}\bigg]^{-1/2}\\
    &= N_0 \ell \bigg[ \frac{1}{c_T^2} + \frac{1}{(1+\ell c_T)^2}\bigg]^{-1/2}\\
    &= N_0 \ell c_{T-1}
  \end{flalign*}
  where $c_{T-1} = \bigg[ \frac{1}{C_T^2} + \frac{1}{1 + \ell c_T}\bigg]^{-1/2}$. Now
  reasoning by recurrence on $t$, assume that $N_{e_{t'}} = N_0 \ell c_{t'}$ with
  $c_{t'} = \bigg[ \frac{1}{c_{t'+1}^2}
  + \frac{1}{(1+\ell \sum_{t''>t'} c_{t''})^2}\bigg]^{-1/2}$ for all $t' \geq t$. We
  then have:
  \begin{flalign*}
    \frac{1}{N_{e_{t-1}}^2}
    &= \frac{1}{N_{e_t}^2}
    + (1-\rho) \frac{1}{(N_0 + \sum_{t'=2}^T \iv\{t' > t-1\} N_{e_{t'}})^2}\\
    &= \frac{1}{N_0^2 \ell^2 c_t^2}
    + (1-\rho) \frac{1}{(N_0 + \sum_{t'=2}^T \iv\{t' > t-1\} N_0 \ell c_{t'})^2}\\
    &= \frac{1}{N_0^2}(1-\rho) \bigg[
    \frac{1}{(1-\rho) \ell^2 c_t^2}
    + \frac{1}{(1 + \sum_{t'=2}^T \iv\{t' > t-1\} \ell c_{t'})^2}\bigg] \\
    &= \frac{1}{N_0^2}(1-\rho) \bigg[
    \frac{1}{c_t^2}
    + \frac{1}{(1 + \sum_{t'=2}^T \iv\{t' > t-1\} \ell c_{t'})^2}\bigg] \\
  \end{flalign*}
  and therefore:
  \begin{flalign*}
    N_{e_{t-1}}^2
    &= N_0 (1-\rho)^{-1/2} \bigg[
    \frac{1}{c_t^2}
    + \frac{1}{(1 + \sum_{t'=2}^T \iv\{t' > t-1\} \ell c_{t'})^2}\bigg]^{-1/2}\\
    &= N_0 \ell c_{t-1}.
  \end{flalign*}
  This proves that $N_{e_t} = N_0 \ell c_t$ for all $t=2, \ldots, T$. Now since all
  units are assigned to exactly one arm, we have:
  \begin{flalign*}
    N_0 + N_1 + \sum_{t=2}^T N_{e_t} = N
    &\quad \Leftrightarrow \quad N_1 = N - N_0 - \sum_{t=2}^T N_{e_t}\\
    &\quad \Leftrightarrow \quad N_1 = N - N_0 \bigg[ 1 + \ell \sum_{t=2}^T c_t\bigg].   
  \end{flalign*}
  The last step is to obtain an expression for $N_0$ as a function of $\{c_t\}_{t=2}^T$
  and $\ell$. Notice that:
  \begin{flalign*}
    \frac{d\rho}{dN_{e_2}} = 0
    &\quad \Leftrightarrow \quad
    \rho \frac{T-1}{(N - N_0 - \sum_{t'=2}N_{e_{t'}})^2} - \frac{1}{N_{e_2}^2} = 0\\
    &\quad \Leftrightarrow \quad
    N_{e_2} = [(T-1)\rho]^{-1/2} \bigg[ N - N_0 ( 1 + \ell\sum_{t'=2}^T c_{t'})\bigg]\\
    &\quad \Leftrightarrow \quad
     [(T-1)\rho]^{1/2} N_{e_2} = N - N_0 ( 1 + \ell\sum_{t'=2}^T c_{t'})\\  
    &\quad \Leftrightarrow \quad
     [(T-1)\rho]^{1/2}N_0 \ell c_2 + N_0 ( 1 + \ell\sum_{t'=2}^T c_{t'}) = N\\
    &\quad \Leftrightarrow \quad
    N_0 \bigg[1 +  [(T-1)\rho]^{1/2} \ell c_2 + \ell \sum_{t'=2}^T c_{t'}\bigg] = N\\
    &\quad \Leftrightarrow \quad
    N_0 = N\bigg[1+[(T-1)\rho]^{1/2} \ell c_2 + \ell \sum_{t'=2}^T c_{t'}\bigg]^{-1}\\
    &\quad \Leftrightarrow \quad
    N_0 = N\bigg[1+([(T-1)\rho]^{1/2} + 1) \ell c_2 +
    \ell \sum_{t'=3}^T c_{t'}\bigg]^{-1}\\   
  \end{flalign*}
  which completes the proof.
\end{proof}

\section{Proof of results in Section~\ref{sec:recycling}}
\label{appendix:recycling}

The proof of Theorem~\ref{th:recycling} follows the same lines as that of
Theorem~\ref{th:minimax-design}, with a few modifications. We first state and prove an
analog to Lemma~\ref{lemma:loss}. Lemmas~\ref{lemma:symmetrization} and \ref{lemma:crd}
carry through unchanged (they only depend on the invariance of the loss). We then
prove a slightly modified version of Lemma~\ref{lemma:max-crd}. The proof of
Theorem~\ref{th:recycling} follows directly from these modified lemmas.
In this section, we use the loss function:
\begin{equation*}
  L(\bZ, \mY) = \sum_{t=2}^T(\hat{\lambda}_t -
  \lambda_t)^2 + \sum_{t=2}^T(\hat{\beta}_t - \delta_t)^2
\end{equation*}
where:
\begin{equation*}
\hat{\beta}_t = \frac{1}{N_{e_t}} \sum_i^N \iv(\bZ_i = \be_t)Y_{it}(\be_t) -
\frac{1}{N_{tk}} \sum_i^N \iv(i \in \Cset_k^{(t)}) Y_{it}(\bZ_i)
\end{equation*}
with $N_{tk} = |\Cset_k^{(t)}|$. Under Assumption~\ref{a:k-recycling}, $\hat{\beta}_t$
can be rewritten:
\begin{equation*}
\hat{\beta}_t = \frac{1}{N_{e_t}} \sum_i^N \iv(\bZ_i = \be_t)Y_{it}(\be_t) -
\frac{1}{N_k^t} \sum_i^N \iv(i \in \Cset_k^{(t)}) Y_{it}(\zero).
\end{equation*}
We now state a slightly modified version of Lemma~\ref{lemma:loss}.
\begin{lemma}[Analog to Lemma~\ref{lemma:loss}]\label{lemma:loss-alt}
	Let $\pi \in \Sset_N$, $\bZ$ and $\mY$. Then:
	\begin{equation*}
		L(\pi\cdot \bZ, \pi \cdot \mY) = L(\bZ, \mY)
	\end{equation*}
\end{lemma}
\begin{proof}
  The only element that is changed from Lemma~\ref{lemma:loss} is the use of
  $\hat{\beta}_t$ instead of $\hat{\delta}_t$. More specifically, under
  Assumption~\ref{a:k-recycling}, if we let:
  \begin{equation*}
    \hat{\beta}_t^{(\be_t)}(\bZ, \mY) = \frac{1}{N_{e_t}} \sum_i^N \iv(\bZ_i = \be_t)Y_{it}(\be_t); 
    \quad
    \hat{\beta}_t^{(k)}(\bZ, \mY) = \frac{1}{N_{tk}} \sum_i^N \iv(i \in \Cset_k^{(t)}) Y_{it}(\zero)
  \end{equation*}
  and $\hat{\beta}_t = \hat{\beta}_t^{(\be_t)}(\bZ, \mY) - \hat{\beta}_t^{(k)}(\bZ, \mY)$, all we need to show 
  is that $\hat{\beta}_t^{(k)}(\pi \cdot \bZ, \pi \cdot \mY) = \hat{\beta}_t^{(k)}(\bZ,\mY)$ 
  for all $\pi$ and all $t=2, \ldots, T$, since this is the only new term.
  
  \medskip
  
  First, notice that:
  \begin{flalign*}
    i \in \Cset_k^{(t)}(\pi \cdot \bZ)  
    &\quad \Leftrightarrow \quad 
    (\pi \cdot \bZ)_i = \zero \,\, \text{ or } \,\,(\pi \cdot \bZ)_i = \be_{t'} \,\, \forall t' \in \Tset_k^{(t)} \\
    &\quad \Leftrightarrow \quad 
    \bZ_{\pi^{-1}(i)} = \zero \,\, \text{ or } \,\, \bZ_{\pi^{-1}(i)} = \be_{t'} \,\, \forall t' \in \Tset_k^{(t)} \\
    &\quad \Leftrightarrow \quad 
    \pi^{-1}(i) \in \Cset_k^{(t)}(\bZ).
  \end{flalign*}
  This implies that
  $\iv(i \in \Cset_k^{(t)}(\pi \cdot \bZ)) = \iv(\pi^{-1}(i) \in \Cset_k^{(t)})$.
  Moreover, since $\pi$ is a permutation, this also implies that:
  \begin{equation*}
    N_{tk}(\bZ)
    = |\Cset_k^{(t)}(\bZ)|
    = |\Cset_k^{(t)}(\pi \cdot \bZ)|
    = N_{tk}(\pi \cdot \bZ)
  \end{equation*}
  and therefore:
  \begin{flalign*}
    \hat{\beta}_t^{(k)}(\pi \cdot \bZ, \pi \cdot \mY) &= 
    \frac{1}{N_{tk}(\pi \cdot \bZ)} \sum_{i\in \bI} \iv(i \in \Cset_k^{(t)}(\pi \cdot \bZ)) (\pi \cdot \bY(\zero))_i^{(t)} \\
    &= \frac{1}{N_{tk}(\bZ)}  \sum_{i\in \bI} \iv(\pi^{-1}(i) \in \Cset_k^{(t)}(\bZ)) Y_{\pi^{-1}(i)}^{(t)}(\zero) \\
    &= \frac{1}{N_{tk}(\bZ)}  \sum_{j\in \pi^{-1} \bI} \iv(j \in \Cset_k^{(t)}(\bZ)) Y_{j}^{(t)}(\zero) \\
    &= \frac{1}{N_{tk}(\bZ)}  \sum_{j\in \bI} \iv(j \in \Cset_k^{(t)}(\bZ)) Y_{j}^{(t)}(\zero) \\
    &= \hat{\beta}_t^{(k)}(\bZ, \mY).
  \end{flalign*}
\end{proof}


\begin{lemma}\label{lemma:max-crd-alt}[Analogous to Lemma~\ref{lemma:max-crd}]
Let $\eta \in \PulsesDesign$ and $\mbY(\Yset)$ where $\Yset$ is bounded and permutation-invariant, then:
\begin{equation*}
	\max_{\mY \in \mbY(\Yset)} \{r(\tilde{\eta}; \mY)\} 
	= V^\ast \sum_{\bZ \in \PulsesSupport} \bigg( \frac{T-1}{N_1(\bZ)}
        + \sum_{t=2}^T \frac{1}{N_{tk}(\bZ)}
	+ 2 \sum_{t=2}^T \frac{1}{N_{e_t}(\bZ)}\bigg)
\end{equation*}
where
\begin{equation*}
	V^\ast = \max_{\bx \in \Yset} \frac{1}{N-1} \sum_{i=1}^N (\bx_i - \bar{\bx})^2
\end{equation*}
\end{lemma}

\begin{proof}
  We only discuss the modifications that need to be made to the proof of
  Lemma~\ref{lemma:max-crd}. 
  
  \medskip
	
  First, for $t$, define
  $\zero^\ast = \{\zero \text{ or } \be_{t'}, \,\, t' \in \Tset_k^{(t)}\}$.
  Since the treatment is completely randomized, then $\zero^\ast$ is also completely
  randomized ($N_{tk}$ out of $N$). This means that
  $\text{bias}_{\tilde{\delta}_{\bZ}}(\hat{\beta}_t, \delta_t)=0$, and:
	\begin{flalign*}
		V_{\tilde{\delta}_{\bZ}} 
		&= \frac{V_{0^\star}^{(t)}}{N_{tk}(\bZ)} + \frac{V_{e_t}^{(t)}}{N_{e_t}(\bZ)} 
		- \frac{V_{0^\ast,e_t}^{(t)}}{N} \\
		&=  \frac{V_0^{(t)}}{N_{tk}(\bZ)} + \frac{V_{e_t}^{(t)}}{N_{e_t}(\bZ)} 
		- \frac{V_{0,e_t}^{(t)}}{N}.
	\end{flalign*}

	The rest of the proof follows the lines of the proof of Lemma~\ref{lemma:max-crd}.
\end{proof}


\begin{proof}[Proof of Theorem~\ref{th:recycling}]
  The proof follows exactly the lines of the proof of
  Theorem~\ref{th:minimax-design}, using Lemmas~\ref{lemma:loss-alt} and
  \ref{lemma:max-crd-alt} instead of Lemmas~\ref{lemma:loss} and \ref{lemma:max-crd}.
\end{proof}


\section{Randomization-inference}
\label{sec:randomization-inference}

In the randomization-based framework we adopt, the potential outcomes are considered
fixed, and the only source of randomness is the random assignment $\bZ$; in particular,
inference is performed with respect to the design $\eta$ used to randomize the
assignment. The minimax optimal designs we obtain lend themselves to straightforward
randomization-based inference since they are completely randomized experiments, and the
estimator used are differences-in-means. In particular, if $\eta^{\opt}$ is a minimax
optimal design as in Theorem~\ref{th:minimax-design}, then 
$\text{Bias}_{\eta^{\opt}}(\hat{\lambda}_t, \lambda_t) = 0$, and
$\text{Bias}_{\eta^{\opt}}(\hat{\delta}_t, \delta_t) = 0$, for $t=2, \ldots, T$. The
usual variance formulas hold \citep{imbens2015causal}:
\begin{equation*}
   \text{Var}_{\eta^{\opt}}(\hat{\lambda}_t) = 
   \frac{V^{(t)}_1}{N_1^{\opt}} + \frac{V^{(t)}_{e_t}}{N_{e_t}^{\opt}}
   - \frac{V^{(t)}_{1e_t}}{N}; \quad 
   \text{Var}_{\eta^{\opt}}(\hat{\delta}_t) = 
   \frac{V^{(t)}_0}{N_0^{\opt}} + \frac{V^{(t)}_{e_t}}{N_{e_t}^{\opt}}
   - \frac{V^{(t)}_{0e_t}}{N}; \qquad
   \text{ for } t=2, \ldots, T
\end{equation*}
where 
\begin{flalign*}
  V_{h}^{(t)} &= \frac{1}{N-1} \sum_{i=1}^N \{Y_{it}(\bh) -
  \overline{Y}_t(\bh)\}^2, \quad h = 1, 0, e_t \\
  V^{(t)}_{he_t} &= \frac{1}{N-1} \sum_i^N 
  \{Y_{it}(\bh) - \overline{Y}_t(\bh))\} \{Y_{it}(\be_t)
  - \overline{Y}_t (\be_t))\}, \quad h = 1, 0
\end{flalign*}
If $\eta^{\opt}$ is a minimax optimal design as in Theorem~\ref{th:partial-recycling},
similar results hold but with
\begin{equation*}
   \text{Var}_{\eta^{\opt}}(\hat{\gamma}_t) = 
   \frac{V^{(t)}_0}{N_t^{\opt}} + \frac{V^{(t)}_{e_t}}{N_{e_t}^{\opt}}
   - \frac{V^{(t)}_{0e_t}}{N}; \qquad
   \text{ for } t=2, \ldots, T  
\end{equation*}
The variances of the estimators under the designs obtained in our other theorems can
be obtained similarly. The standard conservative estimators of these variance elements
can be obtained \citep{imbens2015causal} and since $\hat{\lambda}_t$, $\hat{\delta}_t$
and $\hat{\gamma}_t$ are asymptotically normal under mild conditions, conservative
confidence intervals can be constructed.


\section{Expected risk}
\label{subsection:risk}

Our theory predicts that that our minimax design will minimize the maximum
risk –– the previous set of simulations illustrated the magnitude of the
reduction. However, our theory does not offer any sort of guarantee on the expected
risk under any specific outcome model. The simulations in this section aim to 
give some insights into how our designs perform in terms of expected risk under two simple outcome models. 
\begin{itemize}
\item \texttt{Standard}: Consider the following model:
  \begin{equation}\label{eq:model1}
    Y_{it}(\bZ) = \mu + \alpha_i + \beta_t + Z_{it} \delta + \gamma_{Z_{i(t-1)}} +
    \epsilon_{it},
 \end{equation}
 where $\alpha_i$ and $\beta_t$ are fixed effects associated to unit $i$ and time period
 $t$, respectively. It can be verified that under this model, the expected value of the
 instantaneous effect is constant and equal to  $\mu + \delta$. The term $\gamma_{Z_{i(t-1)}}$
 captures the residual effect from the previous time step. 
\item \texttt{Habituation}: Consider the following model:
  \begin{equation}\label{eq:model3}
    Y_{it}(\bZ) = \mu + \alpha_i + \beta_t + Z_{it}\delta - Z_{it} Z_{i(t-1)} \rho \delta
    + \epsilon_{it},
  \end{equation}
  where $\rho\in[0, 1)$ represents a decay in treatment efficacy if the treatment
  is repeated between successive time periods.
\end{itemize}

We set the parameters values to $\mu=0$, 
$\alpha_i=\log(i)$, $\beta_t =\log(t)$, $\delta=1.0$, $\gamma=-1$, $\rho=0.5$, $\epsilon_{it}\sim \mathcal{N}(0,4^2)$. The results presented below appear to be quite 
robust to different parameter specifications.~\footnote{We also tried
$\alpha_i = \sqrt{i}, \beta_t = \sqrt{t}$, and $\delta$ from 1 to 5 and $\gamma$ from -5 to -1, and their combinations, without seeing any qualitative change in the results.}
We consider settings with $N = 100, 200, 500$ units, $T = 10, 15, 20, 25, 30$ maximum time periods, and perform $100$ runs for every experimental setting. 

Figure~\ref{fig:risk} displays the distribution of the risk values defined in Equation~\eqref{eq:risk} for the minimax and BCRD designs. From the results of Figure~\ref{fig:risk}, we see that the minimax design achieves, in general, smaller risk values than BCRD, especially as the number of units, $N$,
increases. 
For fixed $N$, an increase in $T$ generally leads to an increase in expected risk values. 
This is expected since there are effectively fewer data to estimate the individual
estimands, $\lambda_t, \delta_t$~(the effect is more evident 
in the $N=100$ panel). This also explains why the variance for both designs decreases in
general as $N$ increases. 

\begin{figure}[!t]
	\centering
	\includegraphics[width=1.1\textwidth]{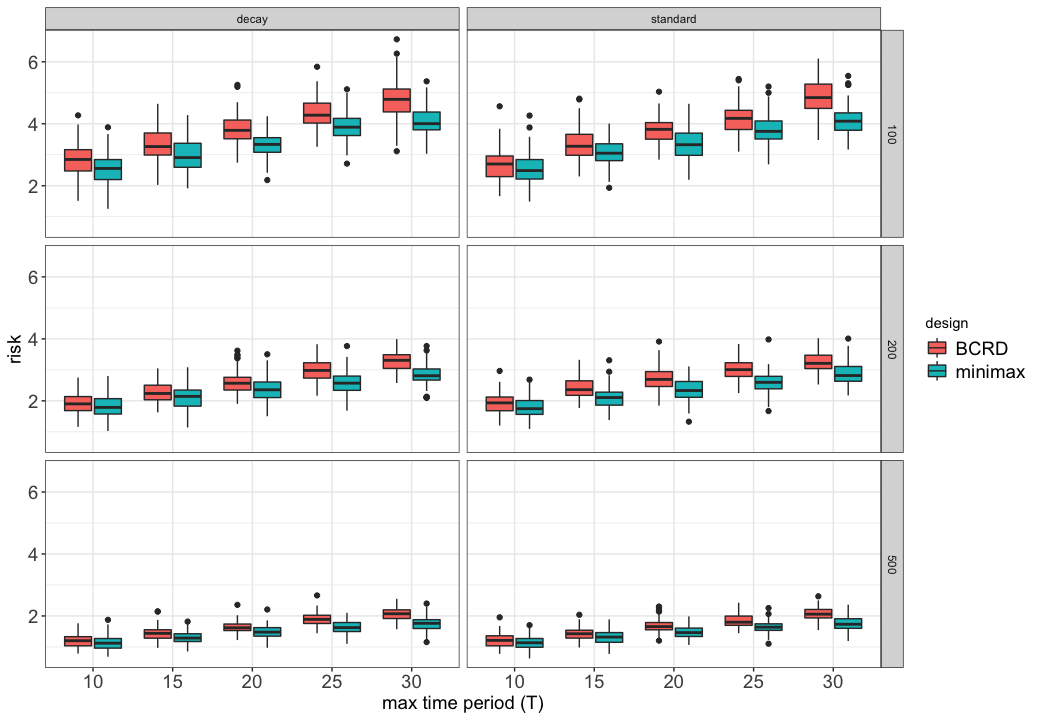}
	\vspace{-0.1in}\caption{Risk of the estimators for the minimax design 
	and the BCRD design with the two outcome models defined in Equation~\eqref{eq:model1}
	and Equation~\eqref{eq:model3}.}
	\label{fig:risk}  
\end{figure}



\end{document}